\documentclass[12pt,a4paper]{article}

\usepackage{latexsym,amsmath,amscd,amssymb,amsthm,graphics}
\usepackage{enumerate,epsfig}
\usepackage[margin=2cm]{geometry}
 \usepackage{float}
\usepackage{graphicx}
\usepackage{framed,url,multicol}
\usepackage[usenames]{color}
\usepackage{shadow,enumerate}
\usepackage{makeidx}

\usepackage[colorlinks]{hyperref}

\usepackage[all]{xy}

\newcommand{\bfi}{\bfseries\itshape}

\makeatletter

\@addtoreset{figure}{section}
\def\thefigure{\thesection.\@arabic\c@figure}
\def\fps@figure{h, t}
\@addtoreset{table}{bsection}
\def\thetable{\thesection.\@arabic\c@table}
\def\fps@table{h, t}
\@addtoreset{equation}{section}

\makeatother

\newcommand{\x}{\xi}

\newcommand{\e}{\eta}

\newcommand{\de}{\delta}

\begin{document}

\newtheorem{theorem}{Theorem}[section]
\newtheorem{definition}[theorem]{Definition}
\newtheorem{lemma}[theorem]{Lemma}
\newtheorem{remark}[theorem]{Remark}
\newtheorem{proposition}[theorem]{Proposition}
\newtheorem{corollary}[theorem]{Corollary}
\newtheorem{example}[theorem]{Example}

\setlength\parindent{0pt}



\title{Toda lattice G-Strands}

\author{
Darryl D. Holm$^{1}$ and Alexander M. Lucas$^{1}$
}
{\addtocounter{footnote}{1} 
\footnotetext{Department of Mathematics, Imperial College, London SW7 2AZ, UK. \\
Email: \texttt{d.holm@imperial.ac.uk, alexander.lucas09@imperial.ac.uk}
}%
%
\date{}

\maketitle

\makeatother
\maketitle


%

\begin{abstract} 
Hamilton's principle is used to extend for the Toda lattice ODEs to systems of PDEs called the Toda lattice strand equations (T-Strands). 
The T-Strands in the $n$-particle Toda case comprise $4n-2$ quadratically  nonlinear PDEs in one space and one time variable. T-Strands form a symmetric hyperbolic Lie-Poisson Hamiltonian system of quadratically nonlinear PDEs with constant characteristic velocities.  The travelling wave solutions for the two-particle T-Strand equations are solved geometrically, and their Lax pair is given to show how nonlinearity affects the solution. The three-particle T-Strands equations are also derived from Hamilton's principle. For both the two-particle and three-particle T-Strand PDEs the determining conditions for the existence of a quadratic zero-curvature relation (ZCR) exactly cancel the nonlinear terms in the PDEs.  Thus, the two-particle and three-particle T-Strand PDEs do not pass the ZCR test for integrability.
\end{abstract} 


\vspace{10 mm}

\tableofcontents

\section{Introduction and plan of the paper}\label{intro-sec}

The concept of G-Strand PDE dynamics applies to motion of molecular strands, or filaments, such as polymers and other long hydrocarbon molecules such as DNA in which a molecular interaction occurs among neighboring points along the strand \cite{ElGBHoPuRa2010}. In what follows, the molecular interactions at each point $s$ along the strand are \textit{simplified} by replacing the physical interaction model with an $n$-particle Toda lattice, coupled to its neighbors along the strand by gradient terms in $\partial_s$. This simplification of the molecular strand interaction yields the Toda lattice G-Strand, abbreviated as \textit{T-Strand}. Our approach uses the Euler-Poincar\'e (EP) formulation of Hamilton's principle for the Toda lattice ODEs \cite{HoMaRa1998}. In the EP formulation, the Toda lattice dynamical system is recognised as coadjoint motion on the dual of a certain matrix Lie algebra. In the EP setting, we pass from the Toda lattice ODEs to a system of nonlinear PDEs in one space and one time dimension. 

\paragraph{Lax pair representation}
The Toda lattice ODEs are well known to admit a matrix commutator representation due to Flaschka \cite{Fl1974a,Fl1974b}
\begin{align}
\frac{dL}{dt} =  [L,\,M] = LM - ML\,,
\label{LPRep}
\end{align}
where $L=L^T$ is symmetric and $M=-M^T$ is antisymmetric. The matrices $L$ and $M$ in \eqref{LPRep} are said to form a \textit{Lax pair}, after Lax's famous solution of the Korteweg--de Vries (KdV) equation \cite{Lax68} obtained by writing KdV in the form \eqref{LPRep}.
\paragraph{Zero curvature representation (ZCR)}
Soon after Lax wrote the KdV equation in commutator form \eqref{LPRep} in \cite{Lax68}, people realized that the Lax-pair representation of integrable systems is equivalent to a zero curvature representation (a zero commutator of two operators), as in 
\begin{equation} 
\partial_t L - \partial_s M = [L,M] 
\label{ZCR-eqn-intro(1)}
\,.\end{equation}
In equation \eqref{ZCR-eqn-intro(1)}, one sees two separate Lax pairs (one in $t$ and another in $s$) whose $(s,t)$ dependence on both independent variables is linked together by the conditions imposed by the ZCR. 

\paragraph{Summary of the paper}
In this paper, we first formulate the Toda ODEs as an Euler-Poincar\'e system using Hamilton's principle on a semidirect-product Lie group of scalings and translations $S\circledS T$.  Next, we derive the T-Strand PDEs by using the Euler-Poincar\'e theory to include another independent variable and thus extend from Toda ODEs in time $t$ to T-Strand PDEs in space-time $(s,t)$. This approach was previously used successfully to formulate integrable PDEs by extending Euler's equations for the rigid body on $SO(3)$ and the Bloch-Iserles equations on $Sp(2)$ from ODEs to PDEs \cite{HoIvPe2012}. 

Next, we test for integrability of the T-Strand PDEs by formulating their zero curvature representation (ZCR). It turns out that the T-Strand systems for two-particle and three-particle T-Strands admit a ZCR only on a submanifold of the full solution space defined by certain linear relationships for which their (quadratic) nonlinear terms cancel. The role of the nonlinear terms for travelling wave ODEs is discussed for the two-particle T-Strands and the Lax pair is given for these travelling waves. 

The main results and plan of the paper are as follows. 
\begin{description}
\item
Section \ref{Toda-sec} briefly introduces the Toda lattice. 

\item
Section \ref{EPToda-sec} begins by formulating Toda lattice dynamics in terms of Hamilton's principle by using the Euler-Poincar\'e (EP) theory for Lie group invariant Lagrangians \cite{Po1901}. This section places Toda lattice dynamics into the EP framework and introduces G-Strands as maps $g(t,s): \mathbb{R}\times\mathbb{R}\to G$ of space-time $\mathbb{R}\times\mathbb{R}$ into a Lie group $G$, such that the maps follow from Hamilton's principle for a $G$-invariant Lagrangian. 

\item
Section \ref{2PT-strand-sec} develops the Lie-Poisson Hamiltonian formulation of the two-particle and three-particle Toda lattice G-Strands. This is the \emph{T-Strand} dynamics that we study. Section \ref{2PT-strand-sec} demonstrates that the two-particle T-Strand is a $4n-2=6$ dimensional symmetric hyperbolic system whose characteristic speeds $c$ are given by $c^2=1$.  All calculations are performed explicitly using elementary matrix methods. 

\item
Section \ref{TWs-sec} discusses the travelling wave solutions of the two-particle T-Strand and derives an explicit solution.  These solutions may be visualized as intersections in $\mathbb{R}^3$ of level sets of two families of quadratic conserved quantities. In particular, these level sets in $\mathbb{R}^3$ comprise two families of orthogonally aligned off-set elliptical cylinders. A typical solution following one of these intersections is shown in Figure \ref{fig1}.

\item Section \ref{ZCR-test} applies the zero curvature representation (ZCR) integrability test to the six-dimensional two-particle T-Strand equations and shows that the determining relations among the variables for which these equations admit a ZCR imply that the nonlinear terms in the T-Strand equations exactly vanish. 

\item Section \ref{3PT-Strand-sec} discusses the ten-dimensional three-particle T-Strands $(4n-2=10)$ and shows the parallel construction of the equations and their ZCR integrability test, which has the same exact cancellation of the nonlinear terms as found for the two-particle T-Strands.
 
 \item
 Section \ref{conclude-sec} concludes with a summary of the main paper's results. One of the main results is the cancellation of the nonlinear terms in the T-Strand equations by the determining relations for them to admit a ZCR. This was not entirely expected, because the present EP approach had previously succeeded in constructing integrable 1+1 PDE generalizations of the integrable ODEs for the Lie algebras $\mathfrak{so}(3)$ (rigid body) and $\mathfrak{sp}(2)$ (Bloch-Iserles equation)  \cite{HoIvPe2012}. For the Toda lattice T-Strands the EP approach produces systems of 1+1 PDEs for which the (linear) determining equations that provide the submanifold on which the system admits a ZCR also happen to cancel the (quadratic) nonlinear terms in these PDEs. This result raises the following question.  What is the mathematical reason why this exact cancellation occurs? One might imagine that it occurs for a Lie-algebraic reason. The EP approach produces equations whose solutions lie on coadjoint orbits of the symmetry group of their Lagrangian in Hamilton's principle. Perhaps the exact cancellation of the nonlinearity in the EP and compatibility equations by their ZCR determining equations has something to do with consistency of the Lie algebraic structure with transforming between ad and ad$^*$ operations of the Lie the symmetry group of the Lagrangian.  An answer for this intriguing question will be conjectured in section \ref{conclude-sec}, but it will not be fully answered, and thus will remain open as a topic for future discussion. 
 \end{description}

\begin{figure}[H]
     \centering
     \includegraphics[width=0.45\textwidth]{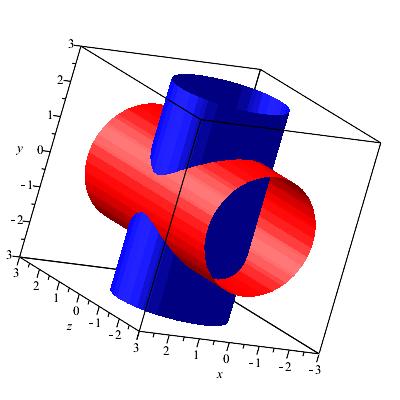}
	\caption{Here is one example of a travelling wave solution for the two-particle T-Strand studied in section \ref{TWs-sec}. These solutions can be represented in $\mathbb{R}^3$ as the intersections of orthogonal  elliptic cylinders. Two solutions are shown, corresponding to periodic motions around the two `pringle' shaped intersections of the elliptic cylinders.}
\label{fig1}
 \end{figure}

\section{The Toda Lattice}\label{Toda-sec}
\subsection{Introduction}
The Toda lattice is a system of $n$ unit masses, connected by nonlinear springs governed by an exponential restoring force \cite{Toda1970}. The canonical equations of motion are derivable from the Hamiltonian 
\begin{align}
H = \sum_{k=1}^n \frac12 p_k^2 + e^{-(q_k-q_{k-1})},
\label{Toda-Ham}
\end{align}
in which $q_k$ is the displacement of the $k$th mass from equilibrium, and $p_k$ is the canonically conjugate momentum. 
This is a special case of the Fermi-Pasta-Ulam (FPU) lattice \cite{FPU1955}, whose Hamiltonian is
\begin{align}
H = \sum_{k=1}^n \frac12 p_k^2 + V(q_k-q_{k-1})
\,.
\label{FPU-Ham}
\end{align}
FPU had focused mainly on the effects of cubic and quartic higher-order terms in the potential $V$, whereas Toda introduced an exponential nonlinearity.

\subsection{Flaschka's change of variables}
\begin{definition}[Flaschka variables]\rm $\,$\\
Flaschka \cite{Fl1974a,Fl1974b} introduced the following $2n-1$ new variables for the $n$-particle Toda Lattice, 
\begin{align}
	a_k := \frac{1}{2}e^{-(q_k - q_{k-1})/2} 
	\quad\hbox{and}\quad
	b_k := -\,\frac{1}{2}p_k .  
\label{Flaschka-var}
\end{align}
\end{definition}

One may check that if $(q_k,p_k)$ satisfy the canonical equations derived from the Hamiltonian in \eqref{Toda-Ham} then $(a_k,b_k)$ satisfy
\begin{align}
\begin{split}
	\dot{a}_k &=  a_k(b_{k+1}-b_k),
	\quad k=1,\dots n-1
	\\
	\dot{b}_1 &=  2a_1^2
	\\
	\dot{b}_k &= 2(a_k^2-a_{k-1}^2),
	\quad k=2,\dots n-1,  
	\\
	\dot{b}_n &= -2a_{n-1}^2
\end{split}
\label{Flaschka-eqns}
\end{align}
with boundary conditions $a_0=0=a_n$. 

The $n$-particle Toda system \eqref{Flaschka-eqns} has a {\bfi Lax pair representation} \cite{Fl1974a,Fl1974b}
\begin{align}
	\frac{dL}{dt} = [M,L] = ML-LM , 
\label{Lax-pair}
\end{align}
for the symmetric tridiagonal matrix $L$ and the antisymmetric tridiagonal matrix $B$ given by
\begin{align}\footnotesize
	L = \begin{pmatrix}
	b_1 & a_1 &   0 & \dots & 0 \\
	a_1 & b_2 & a_2& \dots & 0 \\
	       &        & \ddots        &        & \\
	       &        &        &    b_{n-1}    & a_{n-1} \\
	 0    &        &        &    a_{n-1}    & b_n \\
	\end{pmatrix} 
	\hbox{ and }
	M = L_+ - L_- =
	\begin{pmatrix}
	0 & a_1 &   0 & \dots & 0 \\
	-a_1 & 0 & a_2& \dots & 0 \\
	       &        & \ddots        &        & \\
	       &        &        &    0    & a_{n-1} \\
	 0    &        &        &    -a_{n-1}    & 0 \\
	\end{pmatrix} 
.\label{LB-defs}
\end{align}
\normalsize
Here the subscript $(+)$ in $L_+$ means the upper triangular part of $L$, and $(-)$ in $L_-$ means the lower triangular part. In this case, the Lax pair representation \eqref{LB-defs} means that 
\[
\frac{d}{dt}\big(O(t)^{-1}L(t)O(t)\big) = 0
\,,
\]
for an orthogonal matrix $O(t)\in O(n)$, upon identifying $M=O(t)^{-1}\frac{dO(t)}{dt}$. 

Thus, $O(t)^{-1}L(t))O(t)=L(0)$, so the $n$ eigenvalues of $L(t)$, which are real and distinct,  are preserved along the Toda flow. This is enough to show that the Toda system is an integrable Hamiltonian system \cite{Fl1974a,Fl1974b}.

\paragraph{Other references}
For a broad outline of the mathematics of the Toda lattice and references to some of its extensive literature in the context of integrable Hamiltonian systems, see \cite{AdvMoVa2004}. For a summary introduction to the fundamental papers in the field and a clear discussion of the Toda lattice from the viewpoint of geometric mechanics, see \cite{BlBrRa1990}. The first work on solving the Toda lattice equations by using their integrability is due to Moser \cite{Mo1976}.

\section{Euler--Poincar\'e formulation of Toda lattice dynamics}\label{EPToda-sec}

This section begins by formulating Toda lattice dynamics in terms of Hamilton's principle. For this purpose, we will apply the Euler-Poincar\'e theory \cite{Po1901}. We will then formulate G-Strand dynamics for a strand of interacting Toda lattices depending on space and time.

\subsection{Euler--Poincar\'e Theory}
In the notation for the adjoint (${\rm ad}$) and coadjoint (${\rm ad}^*)$ actions of Lie algebras on themselves and on their duals, Hamilton's principle (that the equations of motion arise from stationarity of the action) for Lagrangians defined on Lie algebras may be expressed as follows. This is the Euler--Poincar\'e theorem \cite{Po1901}. 
\index{Poincar\'e!1901 paper}

\begin{theorem}[Euler--Poincar\'e theorem \cite{Po1901}]
\label{HamPrincLieAlg}$\,$

Stationarity
\begin{equation}\label{stationarity-cond}
\delta S(\xi)=\delta \int^b_a l (\xi)\, d t = 0
\end{equation}
of an action 
\[
S(\xi)=\int^b_a l (\xi)\, d t
\,,
\]
whose Lagrangian is defined on the (left-invariant) Lie algebra $\mathfrak{g}$ of a Lie group $G$ by $l (\xi):\,\mathfrak{g}\mapsto\mathbb{R}$, 
yields the {\bfi Euler--Poincar\'e equation} on $\mathfrak{g}^*$,
\begin{equation}\label{EP-eqns}
\frac{d}{dt} \frac{\delta l}{\delta \xi} = {\rm ad}^*_\xi
\frac{\delta l}{\delta \xi}\,,
\end{equation}
for variations of the left-invariant Lie algebra element \[\xi=g^{-1}\dot{g}(t)\in\mathfrak{g}\] that are restricted to
the form
\begin{equation}
\delta \xi = \dot \eta + {\rm ad}_\xi \, \eta \,,
\label{var-cond}
\end{equation}
in which $\eta(t)\in\mathfrak{g}$ is a curve in the Lie algebra $\mathfrak{g}$ that
vanishes at the endpoints in time. 
\end{theorem}
The Euler--Poincar\'e theorem provides a useful and convenient means of reduction by symmetry of Hamilton's principles for Lagrangians that are invariant under a non-Abelian Lie group. 
For recent discussions and other applications of this theorem, see \cite{HoMaRa1998,HoScSt2009}.

\subsection{Hamilton's principle for the $n$-particle Toda lattice}
The $n$-particle Toda lattice dynamics arises from a Lagrangian defined on $\mathbb{R}^{2n-1}$ that is invariant under a certain Lie group $G$ of scaling $S$ and shear transformations $T$ that has $2n-1$-parameters, in the set $\{\alpha_1,\dots,\alpha_{n-1}, \beta_1,\dots,\beta_n\}$. We denote the Flaschka variables as a row vector in $\mathbb{R}^{2n-1}$, namely  $(a_1,\dots,a_{n-1},b_1,\dots,b_n)$. The Lie group action we consider here is the left action of the semidirect product group $G=S\circledS T$ on $\mathbb{R}^{2n-1}$, namely,
\begin{align} 
\begin{split} 
&S:\quad(a_1,a_2,\dots,a_{n-1})
\to
(e^{\beta_1-\beta_2}a_1,e^{\beta_2-\beta_3}a_2,\dots,e^{\beta_{n-1}-\beta_n}a_{n-1})\,,
\\
&T:\quad(b_1,b_2,\dots,b_n)
\to
(b_1-\alpha_1a_1,b_2+\alpha_1a_1-\alpha_2a_2,\dots,
b_n + \alpha_{n-1}a_{n-1})\,.
\end{split}
\label{Toda-sym-npart}
\end{align}
We will show that the Toda lattice equations in Flaschka variables comprise an Euler-Poincar\'e equation of the form \eqref{EP-eqns} for the semidirect product action in \eqref{Toda-sym-npart} of $G=S\circledS T$ on $\mathbb{R}^{2n-1}$.

\subsection{Three particles}\label{3PHamForm-sec}
The case of $n=3$ is sufficient to illustrate the general idea. In that case, $2n-1=5$ and the Lie group action $G\times \mathbb{R}^5 \to \mathbb{R}^5$ we consider is,
\begin{align} 
\begin{split} 
(a_1,a_2)
&\to
(e^{\beta_1-\beta_2}a_1,e^{\beta_2-\beta_3}a_2)
\\
(b_1,b_2, b_3)
&\to
(b_1-\alpha_1a_1,b_2+\alpha_1a_1-\alpha_2a_2,
b_3 + \alpha_{2}a_{2})
\,.
\end{split}
\label{Toda-sym-3part}
\end{align}
This group representation could be reduced to four dimensions. However, it is convenient to write the Lie group and its Lie algebra in terms of $5=2n-1$ parameters, in order to match the number of degrees of freedom in the Toda equations.

The group element $g_t\in G=S\circledS T$ and its (left) Lie algebra $\xi\in\mathfrak{g}$ for this group action may be represented as a subgroup of the lower triangular matrices. Namely, 
\begin{align} 
g_t
=
\begin{pmatrix}
e^{\beta_1-\beta_2}  & 0   & 0  & 0  & 0
\\
0 & e^{\beta_2-\beta_3}   & 0  & 0  & 0 
\\
- \alpha_1  & 0    & 1    & 0  & 0
\\
\alpha_1  & - \alpha_2  & 0  & 1  & 0
\\
0  & \alpha_2  & 0         & 0  & 1
\end{pmatrix}
\quad\hbox{and}\quad
\xi:= g_t^{-1}\dot{g}_t = 
\begin{pmatrix}
\xi_3-\xi_4  & 0   & 0  & 0  & 0
\\
0 & \xi_4-\xi_5   & 0  & 0  & 0 
\\
- \xi_1  & 0    & 0    & 0  & 0
\\
\xi_1  & - \xi_2  & 0  & 0  & 0
\\
0         & \xi_2  & 0   & 0  & 0
\end{pmatrix}
\label{triang-matrices}
\end{align}
The nonvanishing commutators among the matrix basis elements in this representation are
\begin{align*}
	[\hat{e}_1,\hat{e}_3] = \hat{e}_1 = [\hat{e}_4,\hat{e}_1] \,, \qquad
	[\hat{e}_2,\hat{e}_4] = \hat{e}_2 = [\hat{e}_5,\hat{e}_2]\,,
	\quad\hbox{where}\quad
	\xi=\sum_{j=1}^5\xi_j\hat{e}_j\,.
\end{align*}
This is the expected form of the Lie algebra commutator for a semidirect product. Remarkably, although both groups $S$ and $T$ are Abelian, their semidirect product $G=S\circledS T$ is not Abelian. 

The Lie algebra commutator ${\rm ad}_\xi\eta = [\xi,\eta]$ is given for two Lie algebra elements $\xi$ and $\eta$ by 
\begin{align}
{\rm ad}_\xi\eta = 
\begin{pmatrix}
0  & 0   & 0  & 0  & 0
\\
0   & 0  & 0  & 0  & 0
\\
- \xi_1(\eta_3-\eta_4) 
+ \eta_1 (\xi_3-\xi_4)          & 0    & 0    & 0  & 0
\\
 \xi_1(\eta_3-\eta_4) 
- \eta_1 (\xi_3-\xi_4) & 
- \xi_2(\eta_4-\eta_5) 
+ \eta_2(\xi_4-\xi_5) &  0  & 0  & 0
\\
0         & \xi_2(\eta_4-\eta_5) 
- \eta_2(\xi_4-\xi_5) & 0   & 0  & 0
\end{pmatrix}
\label{ad-commut-3part}
\end{align}

An element of the dual Lie algebra $\mu\in\mathfrak{g}^*$ is represented by the transpose of its corresponding Lie algebra matrix in \eqref{triang-matrices},
\begin{align*}
\mu = 
\begin{pmatrix}
\mu_3 - \mu_4  & 0   & -\mu_1  & \mu_1  & 0
\\
0   & \mu_4-\mu_5   & 0  & -\mu_2  & \mu_2 
\\
0   & 0    & 0    & 0  & 0
\\
0    &  0  & 0  & 0  & 0
\\
0     &  0  & 0   & 0  & 0
\end{pmatrix}
.\end{align*}
This is also the matrix dual, so it provides a non-degenerate pairing $\langle\,\cdot,\, \,\cdot, \rangle: \mathfrak{g}\times \mathfrak{g}^*\to \mathbb{R}$ via the trace pairing of matrices. 

We compute the ad$^*$ operation by taking the dual of ad by the following computation 
\begin{align*}
\langle \mu,\, {\rm ad}_\xi\eta \rangle 
&=
\frac12 {\rm trace} (\mu \, {\rm ad}_\xi\eta)
\\
&=
\mu_1 \, [\xi_1(\eta_3-\eta_4) - \eta_1 (\xi_3-\xi_4)]
+ \mu_2 \, [ \xi_2(\eta_4-\eta_5) - \eta_2(\xi_4-\xi_5) ]
\\
&=
-\, \mu_1(\xi_3-\xi_4)\,\eta_1 
- \, \mu_2(\xi_4-\xi_5)\,\eta_2 
+ \mu_1\eta_1 \,\eta_3
- \, (\mu_1\xi_1-\mu_2\xi_2)\,\eta_4 
-  \, \mu_2\xi_2\,\eta_5
\\
&=
\Big(\!\!- \mu_1(\xi_3-\xi_4)\,, 
-  \mu_2(\xi_4-\xi_5)\,,\,
  \mu_1\xi_1\,,
\mu_2\xi_2 -\mu_1\xi_1\,,
-   \mu_2\xi_2
\Big)
\!\cdot\!
\Big(
\eta_1,\,\eta_2 ,\, \eta_3 ,\, \eta_4 ,\, \eta_5
\Big)^T
\\
&=
\frac12 {\rm trace}\, ({\rm ad}^*_\xi \mu \, \eta)
\\
&=:
\langle {\rm ad}^*_\xi\mu,\, \eta \rangle 
\end{align*}
In matrix form, the formula for ${\rm ad}^*_\xi\mu$ is
\begin{align*}
{\rm ad}^*_\xi\mu = 
\begin{pmatrix}
2\mu_1\xi_1 - \mu_2\xi_2  & 0   & - \mu_1(\xi_4-\xi_3)  &  \mu_1(\xi_4-\xi_3)  & 0
\\
0   & 2\mu_2\xi_2 -\mu_1\xi_1  & 0  & -\mu_2(\xi_5-\xi_4)  & \mu_2(\xi_5-\xi_4) 
\\
0   & 0    & 0    & 0  & 0
\\
0    &  0  & 0  & 0  & 0
\\
0     &  0  & 0   & 0  & 0
\end{pmatrix}
\end{align*}
or
\begin{align*}
\begin{pmatrix}
({\rm ad}^*_\xi\mu)_1 \\ ({\rm ad}^*_\xi\mu)_2 \\ ({\rm ad}^*_\xi\mu)_3 \\ ({\rm ad}^*_\xi\mu)_4 \\ ({\rm ad}^*_\xi\mu)_5
\end{pmatrix} 
= 
\begin{pmatrix}
0 & 0 & -\mu_1 & \mu_1 & 0
\\
0 & 0 &  0 & -\mu_2 & \mu_2
\\
\mu_1 & 0 & 0 & 0 &  0
\\
-\mu_1 & \mu_2 & 0 & 0 &  0
\\
0 & -\mu_2 & 0 & 0 &  0
\end{pmatrix}
\begin{pmatrix}
\xi_1 \\ \xi_2 \\ \xi_3 \\ \xi_4 \\ \xi_5
\end{pmatrix}
\end{align*}

These formulas are the ingredients needed for writing the Toda lattice equation \eqref{EP-eqns} in Euler-Poincar\'e form, namely, as
\begin{align*}
\frac{d\mu}{dt}  = {\rm ad}^*_\xi \mu
\quad\hbox{with}\quad
\mu := \frac{\partial l }{\partial \xi},
\end{align*}
in which $\xi:= g_t^{-1}\dot{g}_t$ for a Lagrangian $l(\xi)$ that is invariant under $S\circledS T$.  In components, this is
\begin{align}
\begin{split}
\dot{\mu}_1 &=  \mu_1(\xi_4-\xi_3)
\,,\\
\dot{\mu}_2 &=  \mu_2(\xi_5-\xi_4)
\,,\\
\dot{\mu}_3 &=  \mu_1\xi_1 
\,,\\
\dot{\mu}_4 &= \mu_2\xi_2 -\mu_1\xi_1
\,,\\
\dot{\mu}_5 &=  -  \mu_2\xi_2
\,.
\end{split}
\label{ad-star-3part}
\end{align} 
After a Legendre transformation to the corresponding Hamiltonian, $h(\mu)$, with $\xi_k=\partial h/\partial \mu_k$ 
and rearrangement into a matrix product form, this set of formulas becomes
\begin{align}
\begin{pmatrix}
\dot{\mu}_1 \\ \dot{\mu}_2 \\ \dot{\mu}_3 \\ \dot{\mu}_4
\\ \dot{\mu}_5
\end{pmatrix}
&=
\begin{pmatrix}
0 & 0 & -\mu_1 & \mu_1 & 0
\\
0 & 0 &  0 & -\mu_2 & \mu_2
\\
\mu_1 & 0 & 0 & 0 &  0
\\
-\mu_1 & \mu_2 & 0 & 0 &  0
\\
0 & -\mu_2 & 0 & 0 &  0
\end{pmatrix}
\begin{pmatrix}
{\partial h}/{\partial \mu_1} \\ {\partial h}/{\partial \mu_2} \\ {\partial h}/{\partial \mu_3} \\ {\partial h}/{\partial \mu_4} \\ {\partial h}/{\partial \mu_5}
\end{pmatrix}
\label{ad-star-n=3}
\end{align} 
Upon identifying $\mu = (\mu_1,\,\mu_2,\,\mu_3,\,\mu_4,\,\mu_5 ) =(a_1,\,a_2,\,b_1,\,b_2,\,b_3 ) $ in Flaschka's notation and  recalling the three-particle Toda Hamiltonian
\[h
=\mu_1^2+\mu_2^2+\frac12(\mu_3^2+\mu_4^2+\mu_5^2)
=a_1^2+a_2^2+\frac12(b_1^2+b_2^2+b_3^2),
\]
we finally recover the EP formulation of the known formulas for three-particle Toda lattice dynamics
\begin{align*}
\begin{pmatrix}
\dot{a}_1 \\ \dot{a}_2 \\ \dot{b}_1 \\ \dot{b}_2 \\ \dot{b}_3
\end{pmatrix}
&=
\begin{pmatrix}
0 & 0 & -a_1 & a_1 & 0
\\
0 & 0 &  0 & -a_2 & a_2
\\
a_1 & 0 & 0 & 0 &  0
\\
-a_1 & a_2 & 0 & 0 &  0
\\
0 & -a_2 & 0 & 0 &  0
\end{pmatrix}
\begin{pmatrix}
2a_1 \\ 2a_2 
\\ b_1 \\ b_2 \\ b_3
\end{pmatrix}
=
\begin{pmatrix}
a_1 (b_2-b_1) \\ a_2 (b_3-b_2)
\\ 2a_1^2 \\ 2(a_2^2-a_1^2) \\ -2a_2^2
\end{pmatrix}.
\end{align*} 
The Euler-Poincar\'e ODEs for the $n$-particle Toda lattice follows the same pattern as the three-particle case. The two-particle case further simplifies the problem. 

\subsection{Two particles}

The Euler-Poincar\'e equations for the two-particle Toda lattice follows by restricting the corresponding formula for $n=3$ in \eqref{ad-star-n=3}.
In the $n=2$ case, we have $2n-1=3$ and the group action $S\circledS T\times \mathbb{R}^3 \to \mathbb{R}^3$ becomes
\begin{align} 
(a,b_1,b_2)\to (e^{\beta_1-\beta_2}a,b_1-\alpha a,b_2+\alpha a)
\,.\label{Toda-sym-2part}
\end{align}

The matrix representations of the group $G$ and its Lie algebra $\mathfrak{g}$ for $n=2$ are given by
\begin{align*} 
g_t
=
\begin{pmatrix}
e^{\beta_1-\beta_2}  & 0   & 0   
\\
- \,\alpha  & 1    & 0
\\
\alpha  & 0  & 1  
\end{pmatrix}
\quad\hbox{and}\quad
\xi:= g_t^{-1}\dot{g}_t |_{t=0}= 
\begin{pmatrix}
\xi_2-\xi_3  & 0   & 0  
\\
- \,\xi_1  & 0    & 0
\\
\xi_1  & 0  & 0
\end{pmatrix}
=
\xi_1\hat{e}_1+\xi_2\hat{e}_2+\xi_3\hat{e}_3
,\end{align*}
where $\xi_1=\dot{\alpha}|_{t=0}$, $\xi_2=\dot{\beta}_1|_{t=0}$ and $\xi_3=\dot{\beta}_2|_{t=0}$.
The commutators among the matrix basis elements in this representation are
\begin{align*}
	[\hat{e}_2,\hat{e}_3] = 0, \qquad
	[\hat{e}_2,\hat{e}_1] = - \, \hat{e}_1,  \qquad
	[\hat{e}_3,\hat{e}_1] = \hat{e}_1.  
\end{align*}

The Lie algebra commutator is given for two Lie algebra elements $\xi=\boldsymbol{\xi}\cdot\mathbf{\hat{e}}$ and $\eta=\boldsymbol{\eta}\cdot\mathbf{\hat{e}}$ by 
\begin{align}
\begin{split}
{\rm ad}_\xi\eta = 
[\xi,\eta] &= 
\big(\xi_1(\eta_2-\eta_3) 
- \eta_1 (\xi_2-\xi_3)\big)
\,\hat{e}_1
\end{split}
\label{ad-xi-eta}
\end{align}
Thus, the components of ${\rm ad}_\xi\eta$ are
\begin{align}
\begin{split}
({\rm ad}_\xi\eta)_1 &=  \xi_1(\eta_2-\eta_3) - \eta_1 (\xi_2-\xi_3)
\,,\\
({\rm ad}_\xi\eta)_2 &=0= ({\rm ad}_\xi\eta)_3
\,.
\end{split}
\label{ad-xi_eta}
\end{align} 

An element of the dual Lie algebra is represented by the transpose matrix,
\begin{align}
\mu = 
\begin{pmatrix}
\mu_2 - \mu_3    & -\mu_1  & \mu_1 
\\
0   & 0    & 0  
\\
0    &  0  & 0  
\end{pmatrix}
=
\mu_1 \hat{e}^1 + \mu_2 \hat{e}^2 + \mu_3 \hat{e}^3
=
\mu_1 \hat{e}_1^T + \mu_2 \hat{e}_2^T + \mu_3 \hat{e}_3^T
\,.\label{mu-rep2}
\end{align}
Thus, for each $\xi\in \mathfrak{g}$, there exists $\mu\in \mathfrak{g}^*$.

We compute the action ${\rm ad}^*:\mathfrak{g}^*\times\mathfrak{g}\to\mathfrak{g}^*$ in matrix notation as follows
\begin{align}
\begin{split}
\langle \mu,\, {\rm ad}_\xi\eta \rangle 
&=
\frac12 {\rm trace} (\mu \, {\rm ad}_\xi\eta)
\\
&=
-\mu_1(\xi_2-\xi_3)\eta_1 + \mu_1\xi_1\eta_2 - \mu_1\xi_1\eta_3 
\\
&=
(
-\mu_1(\xi_2-\xi_3),\, \mu_1\xi_1,\, -\mu_1\xi_1
)
\!\cdot\!
( \eta_1,\,\eta_2 ,\, \eta_3 )^T
\\
&=
(
({\rm ad}^*_\xi\xi)_1,\, ({\rm ad}^*_\xi\xi)_2,\, ({\rm ad}^*_\xi\xi)_3
)
\!\cdot\!
( \eta_1,\,\eta_2 ,\, \eta_3 )^T
\\
&=
\frac12 {\rm trace}\, ({\rm ad}^*_\xi \xi \, \eta)
\\
&=:
\langle {\rm ad}^*_\xi\xi,\, \eta \rangle 
.
\end{split}
\label{adstar-components}
\end{align}
In matrix form, the formula for ${\rm ad}^*_\xi\mu$ is then represented using \eqref{mu-rep2}--\eqref{adstar-components} as
\begin{align}
{\rm ad}^*_\xi\mu = 
\begin{pmatrix}
2\mu_1\xi_1   & \mu_1(\xi_2-\xi_3)  & -\mu_1(\xi_2-\xi_3) 
\\
0   & 0    & 0  
\\
0    &  0  & 0  
\end{pmatrix}
.
\label{adstar-xi-mu}
\end{align}
\begin{remark}\label{ad-dagger-rem}\rm
As a side remark, we observe that on this Lie algebra the quantity $({\rm ad}^*_\xi\mu)^T\in \mathfrak{g}$ \emph{cannot} be written as a matrix commutator with this pairing, since 
\[
({\rm ad}^*_\xi\mu)^T=:{\rm ad}^\dagger_\xi(\mu^T)\ne -{\rm ad}_\xi(\mu^T)
\,.
\]
In contrast, one has ${\rm ad}^\dagger_\xi(\xi^T) = -{\rm ad}_\xi(\xi^T)$ for both of the Lie algebras $\mathfrak{so}(3)$ and $\mathfrak{sp}(2)$, whose corresponding EP equations are completely integrable.
\end{remark}

Formula \eqref{adstar-xi-mu} for ${\rm ad}^*_\xi \mu$ is the ingredient needed for writing the Euler-Poincar\'e equation 
\begin{align*}
\dot{\mu} = {\rm ad}^*_\xi \mu
\quad\hbox{with}\quad
\mu = \frac{\partial l }{\partial \xi},
\end{align*}
in which $\xi:= g_t^{-1}\dot{g}_t$ for a Lagrangian $l(\xi)$ that is invariant under $S\circledS T$.  In components, this is
\begin{align}
\begin{split}
\dot{\mu}_1 = ({\rm ad}^*_\xi\mu)_1 &=  -\mu_1(\xi_2-\xi_3)
\,,\\
\dot{\mu}_2 = ({\rm ad}^*_\xi\mu)_2 &=  \mu_1\xi_1
\,,\\
\dot{\mu}_3 = ({\rm ad}^*_\xi\mu)_3 &= -\mu_1\xi_1
\,.
\end{split}
\label{EP-comp}
\end{align} 
After Legendre transforming to the corresponding Hamiltonian, 
\[
h(\mu) = \langle\,\mu,\, \,\xi \rangle - l(\xi)
\,,
\]
with $\xi_k=\partial h/\partial \mu_k$ 
and rearrangement into a matrix product form, the set of formulas in \eqref{EP-comp} takes the  Lie-Poisson Hamiltonian form, $\dot{\mu}=B(\mu)\frac{\partial h}{\partial \mu}=\{\mu,h\}$ on the dual Lie algebra $\mathfrak{g}^*$, with
\begin{align}
\begin{pmatrix}
\dot{\mu}_1 \\ \dot{\mu}_2 \\ \dot{\mu}_3 
\end{pmatrix}
&=
\begin{pmatrix}
 0 & -\mu_1 & \mu_1 
\\
\mu_1 & 0 &  0 
\\
- \mu_1 & 0 & 0
\end{pmatrix}
\begin{pmatrix}
{\partial h}/{\partial \mu_1} \\ {\partial h}/{\partial \mu_2} \\ {\partial h}/{\partial \mu_3}
\end{pmatrix}
.\label{ad-star-n=2}
\end{align} 
Upon identifying $(\mu_1, \mu_2,\mu_3)= (a,\,b_1,\,b_2) $, this becomes
\begin{align*}
\begin{pmatrix}
\dot{a}  \\ \dot{b}_1 \\ \dot{b}_2 
\end{pmatrix}
&=
\begin{pmatrix}
0 & -a & a
\\
a & 0 &  0 
\\
-a & 0 & 0 
\end{pmatrix}
\begin{pmatrix}
{\partial h}/{\partial a}  
\\ {\partial h}/{\partial b_1} \\ {\partial h}/{\partial b_2} 
\end{pmatrix}
.\end{align*} 
Now by substituting the two-particle Toda Hamiltonian 
\[h=a^2+\frac12(b_1^2+b_2^2)
=\mu_1^2+\frac12(\mu_2^2+\mu_3^2),\]
we may finally write the standard formulas for two-particle Toda lattice dynamics as an Euler-Poincar\'e system
\begin{align*}
\begin{pmatrix}
\dot{a}  \\ \dot{b}_1 \\ \dot{b}_2 
\end{pmatrix}
&=
\begin{pmatrix}
0 & -a & a
\\
a & 0 &  0 
\\
-a & 0 & 0 
\end{pmatrix}
\begin{pmatrix}
2a \\ b_1 \\ b_2 
\end{pmatrix}
=
\begin{pmatrix}
a (b_2-b_1) \\ 2a^2 \\ -2a^2
\end{pmatrix}
.
\end{align*} 

We have now recovered the well-known Lie-Poisson Hamiltonian formulations of the two-particle and three-particle Toda dynamics, by casting the problem as an Euler-Poincar\'e (EP) problem. The methods discussed generalize easily to $n$ particles. For example, $n$-particle Toda dynamics is Hamiltonian with respect to the Lie-Poisson bracket defined on the dual of the Lie algebra for the left action of the semidirect product group $S\circledS T$ on $\mathbb{R}^{2n-1}$. Our goal is to formulate Toda lattice G-Strand dynamics (T-Strands), and for this we will need the EP formulation. We almost have our goal in sight. We know that $n$-particle Toda lattice ODEs represent coadjoint motion for the action of the Lie group $S\circledS T$ on $\mathbb{R}^{2n-1}$.  In particular, we know how to compute the ad$^*$ operation for this semidirect product Lie group, and this ad$^*$ operation will produce the G-Strand equations for a strand of interacting Toda lattices, as explained in the next section.

\subsection{Toda lattice G-Strands for $G=S\circledS T$}

This section formulates the dynamics of a strand of interacting Toda lattices depending on space and time. For this, it introduces G-Strands as a type of map $g(t,s): \mathbb{R}\times\mathbb{R}\to G$ of space-time $\mathbb{R}\times\mathbb{R}$ into a Lie group $G$. 

\begin{definition}[G-Strand \cite{HoIvPe2012}]\rm $\,$\\
A G-Strand is a map $g(t,s): \mathbb{R}\times\mathbb{R}\to G$ of space-time $\mathbb{R}\times\mathbb{R}$ into a Lie group $G$ that follows from Hamilton's principle for a $G$-invariant Lagrangian. 
\end{definition}

Consider Hamilton's principle $\delta S=0$ for a left-invariant Lagrangian, 
\begin{eqnarray}
S=\int_a^b \!\!\!\int_{-\infty}^\infty\!\! \ell(\xi,{\eta})\,ds\,dt
\,,
\end{eqnarray}
with the following definitions of the tangent vectors 
$\xi$ and ${\eta}$,
\begin{eqnarray}
\xi(t,s)=g^{-1}\partial_t g(t,s)
\quad\hbox{and}\quad
{\eta}(t,s)=g^{-1}\partial_s g(t,s)
\,,
\label{xisig-defs}
\end{eqnarray}
where $g(t,s)\in G$ is a real-valued map 
$g:\,\mathbb{R}\times\mathbb{R}\to G$ for a Lie group $G$. 

From equality of cross derivatives, one finds an auxiliary equation for the evolution of ${\eta}(t,s)$
\begin{equation}
\partial_t{\eta}(t,s) - \partial_s\xi(t,s)
= {\eta} \,\xi - \xi\,{\eta}
= [{\eta},\, \xi] 
=: -\,  {\rm ad}_\xi{\eta}
\,.
\label{aux-eqn-2time}
\end{equation}

\paragraph{Hamilton's principle.}
For $\gamma=g^{-1}\delta g(t,s)\in\mathfrak{g}$, Hamilton's principle $\delta S=0$ for 
$
S=\int_a^b \ell(\xi,{\eta})\,dt
$
leads to
\begin{eqnarray*}
\delta S
\!\!\!&=&\!\!\!
\int_a^b
\Big\langle \frac{\delta\ell}{\delta \xi}
\,,\,\delta \xi \Big\rangle
+
\Big\langle \frac{\delta\ell}{\delta {\eta}}
\,,\,\delta {\eta} \Big\rangle
\,dt
\\
\!\!\!&=&\!\!\!
\int_a^b
\Big\langle \frac{\delta\ell}{\delta \xi}
\,,\,\partial_t \gamma + {\rm ad}_\xi\gamma\Big\rangle
+
\Big\langle \frac{\delta\ell}{\delta {\eta}}
\,,\,\partial_s \gamma + {\rm ad}_{\eta}\gamma \Big\rangle
\,dt
\\
\!\!\!&=&\!\!\!
\int_a^b\!
\Big\langle \!-\partial_t \frac{\delta\ell}{\delta \xi}
+ {\rm ad}^*_\xi\frac{\delta\ell}{\delta \xi}
\,,\,\gamma\Big\rangle
+
\Big\langle\! -\partial_s \frac{\delta\ell}{\delta {\eta}}
+ {\rm ad}^*_{\eta}\frac{\delta\ell}{\delta {\eta}}
\,,\,\gamma \Big\rangle
\,dt
\\
\!\!\!&=&\!\!\!
\int_a^b
\Big\langle -\frac{\partial}{\partial t}\frac{\delta\ell}{\delta \xi}
+ {\rm ad}^*_\xi\frac{\delta\ell}{\delta \xi}
- \frac{\partial}{\partial s} \frac{\delta\ell}{\delta {\eta}}
+ {\rm ad}^*_{\eta}\frac{\delta\ell}{\delta {\eta}}
\,,\,\gamma \Big\rangle
\,dt
\,,
\end{eqnarray*}
where the formulas for the variations $\delta \xi$ and $\delta {\eta}$ are obtained from their definitions. Hence, 
$\delta S=0$ yields
\begin{equation}
\frac{\partial}{\partial t} \frac{\delta\ell}{\delta \xi}
- {\rm ad}^*_\xi\frac{\delta\ell}{\delta \xi}
+ \frac{\partial}{\partial s}  \frac{\delta\ell}{\delta {\eta}}
- {\rm ad}^*_{\eta}\frac{\delta\ell}{\delta {\eta}}
=0
\,.
\label{2timeEP1}
\end{equation}
This is the Euler--Poincar\'e equation for the time evolution of $\delta\ell/\delta\xi\in\mathfrak{g}^*$.

Upon defining $\xi:=\delta\ell/\delta \xi$ and $\nu:=\delta\ell/\delta {\eta}$, we have 
\begin{align}
\partial_t\xi - \partial_s\nu
=
{\rm ad}^*_\xi \xi
 - {\rm ad}^*_{\eta} \nu
\quad\hbox{with}\quad
\xi = \frac{\partial l }{\partial \xi}
\quad\hbox{and}\quad
\nu = \frac{\partial l }{\partial {\eta}},
\label{2timeEP2}
\end{align}
in which $\xi:= g^{-1}\partial_t{g}$ and ${\eta}:= g^{-1}\partial_s{g}$ for a Lagrangian $l(\xi,{\eta})$ that is invariant under $G$.

The G-Strand equations also include the auxiliary equation \eqref{aux-eqn-2time} for the evolution of ${\eta}(t,s)$
\begin{equation}
\partial_t{\eta} - \partial_s\xi
= -\,  {\rm ad}_\xi{\eta}
\,,
\label{aux-eqn-xisig}
\end{equation}
obtained from equality of cross derivatives of the definitions \eqref{xisig-defs}.

\begin{remark}\rm
The concept of G-Strands applies to molecular strands, or filaments, such as polymers and other long hydrocarbon molecules such as DNA \cite{ElGBHoPuRa2010}. In what follows, we will simplify the molecular structure at each point along the strand, by replacing the physical model with a two-particle Toda lattice.  This is the two-particle T-Strand. Later in section \ref{3PT-Strand-sec} we will also discuss the three-particle T-Strand, which illustrates the behavior of the $n$-particle case.  
\end{remark}

\section{Two-particle T-Strands} \label{2PT-strand-sec}
This section specializes to the two-particle T-Strand equations and summarizes their Lagrangian and Hamiltonian evolutionary properties.  The wave properties of the two-particle T-Strand equations are revealed by rewriting them equivalently as a six-dimensional symmetric hyperbolic system with constant wave speed. 

The G-Strand Euler-Poincar\'e equation \eqref{2timeEP2} for the case $G=S\circledS T$ in the previous section may be written using equation \eqref{adstar-xi-mu} as
\begin{align}
\begin{split}
\partial_t\xi_1 - \partial_s\nu_1 = ({\rm ad}^*_\xi\xi)_1 - ({\rm ad}^*_{\eta} \nu)_1
&=  -\xi_1(\xi_2-\xi_3) + \nu_1({\eta}_2-{\eta}_3) 
\,,\\
\partial_t\xi_2 - \partial_s\nu_2 = ({\rm ad}^*_\xi\xi)_2 - ({\rm ad}^*_{\eta} \nu)_2 
&=  \xi_1\xi_1 - \nu_1{\eta}_1
\,,\\
\partial_t\xi_3  - \partial_s\nu_3 = ({\rm ad}^*_\xi\xi)_3 -  ({\rm ad}^*_{\eta} \nu)_3
&= -\xi_1\xi_1 + \nu_1{\eta}_1
\,,\end{split}
\label{Gstrand-EP-comp}
\end{align} 
where $\xi=\partial h/\partial \xi$ and ${\eta}=\partial h/\partial \nu$.
The corresponding compatibility equation \eqref{aux-eqn-xisig} may be written using equation \eqref{ad-xi-eta} as
\begin{align}
\begin{split}
\partial_t{\eta}_1 - \partial_s\xi_1 
= - ({\rm ad}_\xi{\eta})_1 
&= -  \xi_1({\eta}_2-{\eta}_3) + {\eta}_1 (\xi_2-\xi_3) 
\,,\\
\partial_t{\eta}_2 - \partial_s\xi_2 
= - ({\rm ad}_\xi{\eta})_2
&=  0
\,,\\
\partial_t{\eta}_3 - \partial_s\xi_3 
= - ({\rm ad}_\xi{\eta})_3
&= 0
\,.\end{split}
\label{Gstrand-aux-comp}
\end{align} 
Equations \eqref{Gstrand-EP-comp}--\eqref{Gstrand-aux-comp} form a system of first order PDEs for $(\x_1 , \x_2 , \x_3 , \e_1 , \e_2 , \e_3)$.

\subsection{Lagrangian and Hamiltonian formulations}

\begin{definition} \rm(Left-invariant G-Strand Lagrangian for the two-particle Toda Lattice)$\,$\\ For $\xi,\eta \in \mathfrak{g}$ we define the left-invariant Lagrangian for the two-particle T-Strand as
\begin{equation} 
l := \left(\frac{1}{2}{\xi_1}^2 +  \frac{1}{2}{\xi_2}^2 + \frac{1}{4}\xi_3^2 \right)
- \left(\frac{1}{2}{\eta_1}^2 +  \frac{1}{2}{\eta_2}^2 + \frac{1}{4}\eta_3^2\right) 
,\label{TS-Lag}
\end{equation}
where we may treat the $\xi$ terms as the kinetic energy of the system and the $\eta$ terms as the potential energy of the system.
\begin{remark} [The choice of the Lagrangian \eqref{TS-Lag}]\rm$\,$\\
Other choices of the Lagrangian would be possible. However, the choice of the T-Strand Lagrangian \eqref{TS-Lag} has been made because it will result in a set of 1+1 PDEs in $(s,t)$ that admits the zero curvature representation, 
\begin{equation} 
\partial_t L - \partial_s M = [L,M] 
\label{ZCR-eqn-Lagchoice}
\,,
\end{equation}
so it contains a Lax pair representation \eqref{Lax-pair} in each independent variable separately. 
\end{remark} 
To obtain the Hamiltonian, $h:  \mathfrak{g}^{\ast}  \to \mathbb{R}$ corresponding to the Lagrangian \eqref{TS-Lag} we take the Legendre transform of this Lagrangian,
\begin{equation*} h(\mu,\e) = \left\langle \mu, \x \right\rangle - l(\xi,\e) \,. \end{equation*}
Taking the total derivative of the Hamiltonian yields
\begin{equation*} d h 
=
 \left\langle d\mu,  \x \right\rangle 
+ \left\langle \mu -\frac{\partial l}{\partial \xi},\,d \xi \right\rangle 
- \left\langle \frac{\partial l}{\partial \e},\,d \e \right\rangle 
,\end{equation*}
from which we find the relations
\begin{equation}
        \frac{\partial h}{\partial \mu} = \x \,,\qquad
        \frac{\partial l}{\partial \x} = \mu \,,\qquad
       \frac{\partial l}{\partial \e} = -\,\frac{\partial h}{\partial \e} \,.
       \label{TS-LegXform}
\end{equation}

Thus, for the two-particle Toda system, the T-Strand Hamiltonian is
\begin{equation*}
h(\mu,\e)= \left(\frac{1}{2}\mu_1^2 +  \frac{1}{2}\mu_2^2 + \mu_3^2 \right)
+ \left(\frac{1}{2}\eta_1^2 +  \frac{1}{2}\eta_2^2 + \frac{1}{4}\eta_3^2  \right)
\,.
\end{equation*}

\end{definition}

\subsection{Space-time evolutionary solutions of the system \eqref{Gstrand-EP-comp}--\eqref{Gstrand-aux-comp}} 
\begin{theorem} \rm \label{Ham-matrix-eqns}
The equation of motion and the compatibility equation may be combined on the Hamiltonian side into the following matrix form,
\begin{equation*}
\frac{\partial}{\partial t} \begin{bmatrix} \mu_1 \\ \mu_2 \\ \mu_3 \\ \e_1 \\ \e_2 \\ \e_3 \end{bmatrix} = 
\begin{bmatrix} 0 & 0 & \mu_3 & \partial_s & 0 & -\e_3 \\ 0 & 0 & -\mu_3 & 0 & \partial_s & \e_3 \\ -\mu_3 & \mu_3 & 0 &0 & 0  & \partial_s -\e_2 + \e_1  \\ \partial_s &0&0&0&0&0 \\ 0&\partial_s&0&0&0&0 \\ \e_3 & -\e_3 & \partial_s + \e_2-\e _1 & 0 & 0 & 0\end{bmatrix} \begin{bmatrix} \mu_1 \\ \mu_2 \\ 2\mu_3 \\ \e_1 \\ \e_2 \\ \frac{1}{2} \e_3 \end{bmatrix}
=: M  \begin{bmatrix} \frac{\de h}{\de \mu} \\  \frac{\de h}{\de \e} \end{bmatrix}
\end{equation*}
\end{theorem}

\begin{proof} 
As a result of Proposition \ref{EP-eqns} we may write the equations of motion and compatibility on the Hamiltonian side as:
\begin{equation} {\partial_t} \mu - {\rm ad}_{ \frac{\delta h}{\delta \mu} }^{\ast}  \mu 
= {\partial_s} \frac{\delta h}{\delta \e}  
- {\rm ad}_{\eta}^{\ast}  \frac{\delta h}{\delta \e} 
\,,\end{equation}
\begin{equation}  
{\partial_t} \eta
- {\partial_s} \frac{\delta h}{\delta \mu} 
= 
-\left[ \frac{\delta h}{\delta \mu},\eta \right] 
= -{\rm ad}_{ \frac{\delta h}{\delta \mu}} \eta 
\,,\end{equation}
which we may formulate in matrix form as
\begin{equation} 
\frac{\partial}{\partial t} \begin{bmatrix} \mu \\ \e \end{bmatrix} = \begin{bmatrix} {\rm ad}_{ \Box }^{\ast} \mu & \partial_s - {\rm ad}_{\eta}^{\ast} \\  \partial_s + {\rm ad}_{\eta}^{\ast}  & 0 \end{bmatrix}  \begin{bmatrix} \frac{\de h}{\de \mu} \\  \frac{\de h}{\de \e} \end{bmatrix} \,.
\end{equation}
The components of ad$^*$ for $(\mathfrak{s}\circledS \mathfrak{t})^*$ are given in equation \eqref{EP-comp} as
\begin{equation*} 
{\rm ad}_{\xi}^{\ast} \mu = \left(\mu_3 \xi^3, -\mu_3 \xi^3, \mu_3 (\xi^2 - \xi^1) \right).
\end{equation*}
Inserting the definitions for coadjoint and adjoint actions yields the result. 
\end{proof}
\begin{remark}\rm The matrix in Theorem \ref{Ham-matrix-eqns} is similar to the lower right corner of the Hamiltonian matrix for a perfect complex fluid \cite{Ho2002,GBRa2009}. It also shows up in the Lie--Poisson brackets for Yang-Mills fluids \cite {HoKu1984} and for spin glasses \cite{DzVo1980,HoKu1988}.
\end{remark}
Finally, we write the two-particle T-Strand equations in component form as
\begin{align}
\begin{split}
&\partial_t\mu_1 - \partial_s\e_1 = 2 \mu_3^2  - \frac{1}{2} \e_3^2\,,  \\
&\partial_t\mu_2 - \partial_s\e_2  = - 2 \mu_3^2 + \frac{1}{2} \e_3^2 \,, \\
& \partial_t\mu_3 -  \partial_s\frac{\e_3}{2}  = - \mu_3(\mu_1 - \mu_2)
+ \frac{\e_3}{2}(\e_1 - \e_2) 
\,, \\
& \partial_t\e_1 - \partial_s\mu_1 = 0 \,, \\
& \partial_t\e_2 - \partial_s\mu_2 = 0 \,, \\
& \partial_t\frac{\e_3}{2} - \partial_s\mu_3 = -\mu_3(\e_1-\e_2 ) + \frac{\e_3}{2}(\mu_1 - \mu_2 ) \,.  
\label{Ham-eqns} 
\end{split}
\end{align}
Inverting the relations 
$\frac{\delta h}{\delta \mu}
=(\mu_1,\mu_2, 2\mu_3)
=\xi= (\xi_1,\xi_2, \xi_3) $
now yields the following.

\begin{theorem}\label{charactform}\rm
Equations \eqref{Ham-eqns} for T-Strands may be written as a six-dimensional symmetric hyperbolic system in $(\x_1 , \x_2 , \x_3 , \e_1 , \e_2 , \e_3)$ with characteristic speeds $c$ given by $c^2=1$.
\end{theorem}

\begin{proof}
The T-Strand equations \eqref{Ham-eqns} may be written in characteristic form, by inverting the Legendre transformation \eqref{TS-LegXform} which gives the linear relations
$ 
(\mu_1,\mu_2,\mu_3) = \left(\x_1,\x_2,\frac{1}{2}\x_3\right) 
,
$ 
and these relations allow rearrangement of \eqref{Ham-eqns} into the following symmetric hyperbolic form,
\begin{equation}
\frac{\partial}{\partial t} 
\begin{bmatrix} \x_1 \\ \x_2 \\ \x_3 \\ \e_1 \\ \e_2 \\ \e_3 \end{bmatrix} 
- 
\begin{bmatrix} 
   0 & 0 & 0 & 1 & 0 & 0 
\\ 0 & 0 & 0 & 0 & 1 & 0 
\\ 0 & 0 & 0 & 0 & 0  & 1 
\\ 1 & 0 & 0 & 0 & 0 & 0 
\\ 0 & 1 & 0 & 0 & 0 & 0 
\\ 0 & 0 & 1 & 0 & 0 & 0
\end{bmatrix} 
\partial_s
\begin{bmatrix} \x_1 \\ \x_2 \\ \x_3 \\ \e_1 \\ \e_2 \\ \e_3 \end{bmatrix}
=
\begin{bmatrix}
  \frac12(\x_3^2 - \e_3^2) 
\\  - \frac12(\x_3^2 - \e_3^2) 
\\ (\x_1-\x_2)\x_3 - (\e_1-\e_2) \e_3
\\ 0 
\\ 0 
\\ (\x_1-\x_2)\e_3 - (\e_1-\e_2)\x_3 
\end{bmatrix}
\label{hyperbolicLag-eqns}
\end{equation}
The characteristic polynomial of this system has roots given by $(c^2-1)^3$=0.
Thus, the characteristic speeds $c$ are given by $c^2=1$.
\end{proof}

\section{Travelling waves for the two-particle T-Strand}\label{TWs-sec}

As shown in the previous section, equations \eqref{Ham-eqns} for the two-particle T-Strand form a symmetric hyperbolic system with constant characteristic speeds $c$ given by $c^2=1$. This is clear from equation \eqref{hyperbolicLag-eqns}, which was obtained by inverting the Legendre transformation \eqref{TS-LegXform} in the proof of Theorem \ref{charactform}. The present section constructs explicit formulas for travelling wave solutions of the system \eqref{hyperbolicLag-eqns} for an arbitrary wave speed, $c\ne\pm1$. 

The space-time translation invariance of equations \eqref{hyperbolicLag-eqns} implies that the system admits the travelling-wave reduction,
\[
\xi_i(s,t)=\bar{\xi}_i(\tau)
\quad\hbox{and}\quad
\eta_i(s,t)=\bar{\eta}_i(\tau),
\quad\hbox{with}\quad
\tau:=s-ct
\quad\hbox{for}\quad
i=1,2,3. 
\]
Substituting these travelling-wave forms of the solutions into equations \eqref{hyperbolicLag-eqns} yields a six-dimensional system of ordinary differential equations (ODE) (dropping the bars for simpler notation)%
\begin{align}
\begin{split}
&-c\dot\x_1 -  \frac{1}{2} \x_3^2 = \dot\e_1 - \frac{1}{2} \e_3^2\,,  \\
&-c\dot\x_2 + \frac{1}{2} \x_3^2 = \dot\e_2 + \frac{1}{2} \e_3^2 \,, \\
& -c\dot\x_3 - \x_3(\x_2 - \x_1) = \dot\e_3 - \e_3(\e_2 - \e_1) \,, \\
& -c\dot\e_1 - \dot\x_1 = 0 \,, \\
& -c\dot\e_2 - \dot\x_2 = 0 \,, \\
& -c\dot\e_3 - \dot\x_3 = - \e_1 \x_3 + \e_3 \x_1 - \e_3 \x_2 + \e_2 \x_3 \,,
\label{TWSLag-eqns}
\end{split}
\end{align}
where one denotes $\dot\x_1=d\x_1 / d\tau$, etc.

This ODE system has three conservation laws that are expressible as invertible linear relations
\begin{align}
\begin{split}
& d_1 :=  c\e_1 + \x_1 = const\,,  \\
& d_2 := c\e_2 + \x_2 = const\,,  \\
& d_3 := c(\x_1+\x_2)+\eta_1+\eta_2 = const \,.
\end{split}
\label{TWSConstants}
\end{align}
Consequently, one may elimate three equations of (\ref{TWSLag-eqns}) in favour of the constants $d_1,d_2$ and $d_3$ in (\ref{TWSConstants}). A series of linear transformations then reduces the travelling wave equations (\ref{TWSLag-eqns}) to the form
\begin{align}
\begin{split}
\dot{X} &=\frac{YZ}{(1-c^2)^2}\,, \\
\dot{Y} &=-Z\left(X + \frac{\alpha}{1-c}\right) , \\
\dot{Z} &=-Y\left(X - \frac{\alpha}{1+c}\right) ,
\end{split}
\label{RedTWS}
\end{align}
where $\alpha$ is a constant of the motion.
\begin{proposition} 
The system \eqref{RedTWS} is divergence free and its solutions can be interpreted as motion along the intersections of two level sets, $C_i = C_i(X,Y,Z)$, for $i=1,2$ in $\mathbb{R}^3$. 
\end{proposition}
\begin{proof}
The proof begins by observing that the system (\ref{RedTWS}) is divergence free in $\mathbb{R}^3$. That is, for a vector ${\bf X}:=(X,Y,Z)^T\in\mathbb{R}^3$ system (\ref{RedTWS}) satisfies
\begin{equation*} 
\nabla\cdot\dot{\textbf{X}}=0\,. 
\end{equation*}
Consequently, the motion preserves the volume form $dX \wedge dY \wedge dZ$, and the system \eqref{RedTWS} can be expressed locally as the cross product of gradients,
\begin{equation}
\dot{\textbf{X}}  = \nabla C_1 \times \nabla C_2 
\,.
\label{TW-vec-eqn}
\end{equation}
Consequently, the solutions of \eqref{RedTWS} move along the intersections of level surfaces of $C_1$ and $C_2$ in $\mathbb{R}^3$.
For our system, this is satisfied for
\begin{align}
\begin{split}
C_1 &=\frac{{(1-c)}^2}{2}\left(X+\frac{\alpha}{1-c}\right)^2+\frac{Y^2}{2{(1+c)}^2} \,, \\
C_2 &=\frac{{(1+c)}^2}{2}\left(X-\frac{\alpha}{1+c}\right)^2+\frac{Z^2}{2{(1-c)}^2}  \,.
\end{split}
\label{2ellipcyls}
\end{align}
The quadratic level surfaces of $C_1$ and $C_2$ comprise two families of elliptic cylinders that are orthogonal to each other in $\mathbb{R}^3$, and whose centres are offset from each other along the $X$-coordinate axis by an amount proportional to $\alpha$.
\end{proof}
\begin{remark} \rm The solution of the travelling wave system (\ref{RedTWS}) has now been reduced to determining the motion along intersections of two families of orthogonal elliptic cylinders in $\mathbb{R}^3$. These are the level sets of $C_1$, whose axis is given as
\begin{equation*} Y=0 \qquad X=-\,\frac{\alpha}{1-c} \,, \end{equation*}
and the level set of $C_2$ whose axis is given as
\begin{equation*} Z=0 \qquad X=\frac{\alpha}{1+c} \,. \end{equation*}
\end{remark}
We note some additional properties of the travelling wave system  \eqref{TWSLag-eqns} that arise from its equivalent representation as \eqref{RedTWS}:
\begin{itemize} 
\item No solutions of \eqref{TWSLag-eqns} exist when $c=\pm1$,  as the elliptic cylinders in \eqref{2ellipcyls} will be undefined at these values. For $c=\pm1$ it can also be shown that (\ref{TWSLag-eqns}) results in an under-determined system of ODEs. This is in stark contrast to the characteristic speeds of the hyperbolic PDE system (\ref{Ham-eqns}), which are given by $c^2=1$. Consequently, one may think of the travelling wave solutions of (\ref{TWSLag-eqns}) and the characteristic solutions of (\ref{Ham-eqns}) as being disjoint solution sets whose union encompasses all possible wave speeds. 
\item Travelling wave solutions of (\ref{TWSLag-eqns}) lie on paths that take the form of the edge of a deformed `pringle' shape, as shown in Figure \ref{fig1}, which illustrates the intersections of level surfaces of $C_1$ and $C_2$. All solutions will therefore be periodic, unless the intersections are critical points. At critical points the two elliptic cylinders are tangent, the gradients of $C_1$ and $C_2$ are collinear, and the travelling wave solution will be stationary in the moving frame.
\end{itemize}
\begin{proposition} $\,$\\
On level sets of $C_1$, for $\alpha=0$, the motion reduces to that of a simple pendulum. \end{proposition}
\begin{proof} For the proof, let us first consider the case where the offset parameter $\alpha$ does not necessarily vanish. Let us define elliptical polar coordinate variables $r$, $\theta$ and $p$, given by
\begin{align}
\begin{split}
&X:=\frac{ r\sin\theta-\alpha}{1-c}\,, \\
&Y:=  r\cos\theta(1+c) \,, \\
&Z:= p(1-c) \,, \\
\end{split}
\label{ellipcoords}
\end{align}
so $C_1$ and $C_2$ transform to
\begin{align}
C_1 = \frac{r^2}{2}  \quad\hbox{and}\quad
C_2= \frac{(1+c)}{2(1-c)}\left( r(1+c) \sin\theta - 2\alpha \right)^2 + \frac{p^2}{2} \,.
\label{C1C2-def}
\end{align}
We restrict the motion to a level set of $C_1$, specified as $2C_1=r^2 = const$. On this level set, the transformation of coordinates in \eqref{ellipcoords} yields
\begin{align}
\begin{split}
&dX= \frac{1}{1-c} ( \sin\theta \,dr + r \cos\theta \,d\theta ),  \\
&dY= (1+c) ( \cos\theta\, dr - r \sin\theta \,d\theta),  \\
&dZ= (1-c)\,dp  
\,. 
\end{split}
\end{align}
The preserved volume form $dX \wedge dY \wedge dZ$ may then be written as
\begin{align}
dX \wedge dY \wedge dZ &= (1+c)\,d\left(\frac{r^2}{2}\right) \wedge d\theta \wedge dp 
= (1+c)\,dC_1 \wedge d\theta \wedge dp 
\,.
\end{align}
This means that the motion on a level set of $C_1$ is canonically Hamiltonian. We choose $C_2$ to be our Hamiltonian and restrict to motion on $C_1$. For a fixed value of $C_1$, the function $\{F,C_2\}$ only depends on $(\theta,p)$, so motions on level sets of $C_1$ are given as
\begin{equation*}
 \{F,C_2\} dX \wedge dY \wedge dZ = dC_1 \wedge  \{F,C_2\}\,(1+c) \, d\theta \wedge dp \,.
\end{equation*}
Hence, motion on the level set $C_1$ is given by the symplectic Poisson bracket,
\begin{equation*}
\{F,C_2\}=\frac{1}{1+c} \left( \frac{\partial F}{\partial \theta} \frac{\partial C_2}{\partial p} - \frac{\partial C_2}{\partial \theta} \frac{\partial F}{\partial p} \right) 
,
\end{equation*}
for which we see
\begin{equation*} \{\theta,p\} = \frac{1}{1+c}\,.\end{equation*} 
When $C_2$ in \eqref{C1C2-def} is chosen as the Hamiltonian, the canonical equations of motion are given by
\begin{align}
\begin{split}
&\frac{d\theta}{d\tau} =  \{\theta,C_2\} = \frac{1}{1+c}  \frac{\partial C_2}{\partial p} = \frac{p}{1+c} \\
&\frac{dp}{d\tau} =  \{p,C_2\} = -\frac{1}{1+c}  \frac{\partial C_2}{\partial \theta} = -\frac{r\cos \theta(1+c)}{(1-c)}(r\sin\theta \,(1+c) - 2\alpha) 
\end{split}
\end{align}
Combining these equations to eliminate $p$ yields
\begin{equation}
\frac{d^2\theta}{d\tau^2} = \frac{1}{(1+c)} \frac{dp}{d\tau} = -\frac{r\cos \theta}{(1-c)}(r\sin\theta \,(1+c) - 2\alpha) \,.
\label{pringle-motion}
\end{equation}
In analogy to rigid body motion \cite{Ho2011}, setting $\alpha=0$ in \eqref{pringle-motion} now yields 
\begin{equation}
\frac{d^2\theta}{d\tau^2} = -\,\frac{C_1(1+c)  }{2(1-c)} \,\sin 2\theta \,,
\label{simplepend-motion}
\end{equation}
which is the equation for a simple pendulum.

\end{proof}

\begin{remark}\rm
For arbitrary values of the parameter $\alpha$, equation \eqref{pringle-motion} describes the motion in $\mathbb{R}^3$ along the `pringle' shape formed at the intersection of the orthogonal elliptic cylinders that represent to two quadratic conserved quantities, as shown in Figure \ref{fig1}.
\end{remark}

\paragraph{Lax pair for two-particle T-Strand travelling waves}$\,$\\
We make one final remark about the two-particle T-Strand travelling waves, in the form of a theorem.
\begin{theorem}\rm
The two-particle T-Strand travelling wave equations may be cast as a Lax pair. 
\end{theorem}

\begin{proof}
We write the travelling wave equations \eqref{TW-vec-eqn} as a 2-form equation 
\begin{align}
\begin{split}
\frac{d\mathbf{X}}{d\tau}\cdot d\mathbf{S}  &= d C_1 \wedge d C_2 
\\&= (1+c)^{-2}d\big( (1+c)^2C_1 + (1-c^2)^2C_2 \big)\wedge d C_2
\\&= \frac{1}{2}(1+c)^{-2}\big( 2(1-c^2)^2XdX + YdY + ZdZ \big)
\\& \hspace{1cm}\wedge d \left((1-c^2)^2\left(X-\frac{\alpha}{1-c}\right)dX 
+ \frac{1}{(1-c)^2} ZdZ \right)
\,,
\end{split}
\label{2formTWcalc}
\end{align}
Defining a new time variable $d\tilde{\tau} = d\tau/(1+c)^2$ and rescaling $X$ to $\widetilde{X}=\sqrt{2}\,(1-c)^2X$ yields
\begin{align}
\mathbf{\widetilde{X}} = (\tilde{X},Y,Z)^T
\quad\hbox{and}\quad
\frac{d\mathbf{\widetilde{X}}}{d\tilde{\tau}}\cdot d\mathbf{\widetilde{S}} 
= 
\sqrt{2}\,(1-c^2)^2\,\frac{d\mathbf{X}}{d\tau}\cdot d\mathbf{S} 
\label{X-tilde-vec}
\end{align}
Consequently, we may write with $\sigma=\sqrt{2}\,(1-c^2)^2\tilde{\tau}$,
\begin{align}
\frac{1}{\sqrt{2}\,(1-c^2)^2}
\frac{d\mathbf{\widetilde{X}}}{d\tilde{\tau}}\cdot d\mathbf{\widetilde{S}} 
=
\frac{d\mathbf{\widetilde{X}}}{d\sigma}\cdot d\mathbf{\widetilde{S}} 
=
\mathbf{\widetilde{X}}\cdot d\mathbf{\widetilde{X}}
\wedge dC_2(\mathbf{\widetilde{X}})
\label{X-tilde-dot}
\end{align}
which in $\mathbf{R}^3$ vector form becomes 
\begin{align}
\frac{d\mathbf{\widetilde{X}}}{d\sigma} 
=
\mathbf{\widetilde{X}}\times
 \widetilde{\nabla}C_2(\mathbf{\widetilde{X}})
\label{X-tilde-vecdot}
\end{align}
Identifying $\mathbf{\widetilde{X}}$ and $ \widetilde{\nabla}C_2$ with the $3\times3$ skew symmetric matrices by the hat map expressions $\widehat{X}_{jk}=-\epsilon_{jki}\widetilde{X}_i$ and  $\widehat{Y}_{jk}=-\epsilon_{jki}\widetilde{\nabla}_iC_2$ then transforms the vector equation \eqref{X-tilde-vecdot} for T-Strand travelling wave motion into Lax commutator form
\begin{align}
\frac{d\widehat{X}}{dt} =  [\widehat{X},\,\widehat{Y}] \,.
\label{LPR-ep}
\end{align}
This is the Lax pair for two-particle T-Strand travelling waves.
\end{proof}

\section{Zero curvature representation (ZCR)}\label{ZCR-test}
Soon after Lax wrote the KdV equation in commutator form in \cite{Lax68}, people realized that the Lax-pair representation of integrable systems is equivalent to a zero curvature representation (a zero commutator of two operators), as in 
\begin{equation} 
\partial_t L - \partial_s M = [L,M] 
\label{ZCR-eqn-intro1}
\,.\end{equation}
In equation \eqref{ZCR-eqn-intro1}, one sees two separate Lax pairs (one in $t$ and another in $s$) whose $(s,t)$ dependence on both independent variables is linked together by the conditions imposed by the ZCR. 
In one of the first applications of the ZCR approach, Zakharov and Manakov \cite{ZaMa1973} developed the inverse scattering solution for the three-wave equation which has a $3 \times 3$ matrix ZCR representation, where both operators $L$ and $M$ are linear in the spectral parameter $\lambda$. Manakov used a particular $s$-independent case of the same ZCR \cite{Ma1976} to find the integrability conditions for the ordinary differential equations (ODEs) describing the motion of a rigid body in $n$ dimensions. The Lax matrices $L$ and $M$ that we study for T-Strands are quadratic in $\lambda$. Lax operators that are polynomial (quadratic) in $\lambda$  were first used in studies of the Thirring model in \cite{KuMi77} and the derivative nonlinear Schr\"odinger (DNLS) equation in \cite{KaNe78}. The inverse scattering method (ISM) for Lax operators that are quadratic in $\lambda$ was developed in \cite{GeIvKu80} and recently a Riemann-Hilbert formulation of the inverse-scattering problem for ZCR operators that are polynomial in $\lambda$ was developed in \cite{Ge12}.

\subsection{The ZCR test for integrability of the two-particle T-Strand} 
Based on the forms of the Lax pair matrices for the Toda lattice in \eqref{LB-defs} and the symmetric form of the two-particle Toda lattice G-Strand system \eqref{Ham-eqns} one may prove the following. 
\begin{proposition}\rm
The two-particle Toda lattice G-Strand equations \eqref{Ham-eqns} possess a zero curvature representation,
\begin{equation}
\partial_t M_1 + \left[M_2,M_1\right]  -\partial_s N_1 - \left[N_2,N_1\right] = 0
\,,\label{EP-eqn}
\end{equation}
where
\begin{equation} 
M_1:= \begin{pmatrix} \xi_1 & \xi_3 \\ \xi_3 & \xi_2 \end{pmatrix}
\quad 
M_2:= \begin{pmatrix} \xi_1 & 0 \\ 2\xi_3  &  \xi_2 \end{pmatrix} 
\quad 
N_1:= \begin{pmatrix} \e_1 & \frac{1}{2}\e_3 \\ \frac{1}{2}\e_3 & \e_2 \end{pmatrix}
\quad 
N_2:= \begin{pmatrix} \e_1 & 0 \\ \e_3 &\e_2 \end{pmatrix} 
\label{Lax-matrices}
.\end{equation}
\end{proposition}
\begin{remark}\rm
Equations \eqref{EP-eqn}--\eqref{Lax-matrices} recover two copies of the Lax pair for the ODE formulation, appearing in $L$ and $M$, if for $M_2$ and $N_2$ we take the matrix differences
\begin{equation*} 
M_2 - M_1 = \begin{pmatrix} 0 & -\xi_3 \\ \xi_3 & 0 \end{pmatrix} 
\quad \hbox{and} \quad
N_2 - N_1 = \begin{pmatrix} 0 & -\,\frac{1}{2}\e_3 \\ \frac{1}{2}\e_3 & 0 
\end{pmatrix}.
\end{equation*}
This recovers antisymmetry of the Lax matrix $M$ for the Toda ODEs. 
Substituting these matrix differences does not affect the reduction to the matrix Lax pair ODEs, but it does affect the additional compatibility equation that arises for the T-Strands. Of course, it  is absent for the original Toda lattice. This is the key step in writing the ZCR for the T-Strands so that it recovers the Lax pair for the Toda ODEs. 
\end{remark}

\begin{proposition} [Lax pair for the compatibility condition]\rm $\,$
\\
It follows from the definitions of the Lax matrices in \eqref{Lax-matrices} that the compatibility condition for two-particle T-Strands implies,
\begin{equation} 
\partial_t N_2 - \partial_s M_2 + \left[M_2,N_2\right] = 0 
\,.
\label{compat-eqn}
\end{equation}
\end{proposition}

\begin{theorem} \rm [ZCR test for integrability of the two-particle T-Strand]$\,$ \label{ZCR-thm}
\\
The T-Strand system of equations comprising the Euler--Poincar\'e equation \eqref{EP-eqn} and the compatibility condition \eqref{compat-eqn}, admit a zero curvature representation,
\begin{equation} 
\partial_t L - \partial_s M = [L,M] 
\label{ZCR-eqn}
\,,\end{equation}
upon choosing
\begin{equation} 
L:=\lambda^2 A + \lambda M_1 - N_2 
\quad  \hbox{and} \quad
M:=\lambda^2 B - \lambda N_1 + M_2 
\,,
\label{Lax-matricesLM}
\end{equation}
where the Lax matrices $(M_1,M_2,N_1,N_2)$ are defined in \eqref{Lax-matrices}.

\end{theorem}

\begin{proof} 
Inserting the definitions for L and M, then equating coefficients of powers of $\lambda$ yields:
\begin{align}
\begin{split}
\lambda^4 &: [A,B]=0 \\
\lambda^3 &: [A,N_1]+[B,M_1]=0 \\
\lambda^2 &: [A,M_2]+[B,N_2]-[M_1,N_1]=0 \\
\lambda^1 &: \partial_t M_1 -\partial_s N_1 + \left[M_2,M_1\right] - \left[N_2,N_1\right] = 0 \\
\lambda^0 &:  \partial_t N_2 - \partial_s M_2 + \left[M_2,N_2\right] = 0
\end{split}
\label{constraint-rels}
\end{align}
We may now solve the determining equations \eqref{constraint-rels} and extract the conditions under which the two-particle Toda G-Strand will possess a zero curvature representation (Lax pair) and thus be an integrable system. 
\begin{itemize}
\item
The equations at order $\lambda^0$ and $\lambda^1$ in the system \eqref{constraint-rels} recover, respectively, the compatibility condition and Euler-Poincar\'e equations. These equations determine the time evolution of $N_2$ and $M_1$ in $L$, which are \emph{prognostic} variables. In contrast, the matrices $A$ and $B$, as well as the \emph{diagnostic} variables $N_1$ and $M_2$ in $M$ must be determined from the ZCR conditions at orders $\lambda^2$, $\lambda^3$ and $\lambda^4$.
\item
The equation at order $\lambda^4$ may be satisfied by choosing the matrices $A$ and $B$ to be constant elements of $SO(2)$; that is, by setting:
\begin{equation} 
A = \begin{pmatrix} a_1 & 0 \\ 0 & a_2 \end{pmatrix} 
\qquad  
B = \begin{pmatrix} b_1 & 0 \\ 0 & b_2 \end{pmatrix}
 \,, \end{equation}
where these constant matrix entries should not be confused with Flaschka's variables appearing in the earlier sections of the paper.
\item
The order $\lambda^3$ equation gives a linear relation that determines the diagnostic variable$N_1$ in terms of the prognostic variable $M_1$, namely,
\begin{align}
N_1 = -\, {\rm ad}_A^{-1}{\rm ad}_B M_1
\label{lambda3-ad}
\end{align}
\item
The order $\lambda^2$ equation gives a formula for the diagnostic variable $M_2$ in terms of the the prognostic variables $N_2$ and $M_1$, namely,
\begin{align}
M_2 
= -\,{\rm ad}_A^{-1}\left( {\rm ad}_B N_2 - {\rm ad}_{M_1}N_1\right)
= -\,{\rm ad}_A^{-1}\left( {\rm ad}_B N_2 
+ {\rm ad}_{M_1}{\rm ad}_A^{-1}{\rm ad}_B M_1\right)
.\label{lambda2-ad}
\end{align}
\end{itemize}
Thus,  the determining equations \eqref{constraint-rels} may be solved to express the diagnostic variables ($N_1,M_2$) contained in $M$,  in terms of the prognostic ones ($N_2,M_1$) contained in $L$. 

Consequently, the $G$-Strand system of equations \eqref{EP-eqn}--\eqref{compat-eqn} may be written as a zero curvature representation (or Lax pair) as in \eqref{ZCR-eqn}, by taking the matrices $L$ and $M$ as
\begin{align} 
\begin{split}
&L=\lambda^2 A + \lambda M_1 - N_2
\,,\\&
M = \lambda^2 B 
+ \lambda\, {\rm ad}_A^{-1}{\rm ad}_B M_1 
-  {\rm ad}_A^{-1}\left( {\rm ad}_B N_2 
+ {\rm ad}_{M_1}{\rm ad}_A^{-1}{\rm ad}_BM_1\right) 
.\end{split}
\label{LM-ops}
\end{align}
\end{proof}

\subsection{A direct solution for the ZCR of the two-particle T-Strand}
The proof of Theorem \ref{ZCR-thm} shows that the relations \eqref{lambda3-ad}--\eqref{lambda2-ad} allow the G-Strand system of equations \eqref{EP-eqn}--\eqref{compat-eqn} to be written as a zero curvature representation (or Lax pair), as in \eqref{ZCR-eqn}, with the matrices $L$ and $M$ given in \eqref{LM-ops}. 
In the present case, the $2\times2$ matrices involved are simple enough to allow a direct solution of the determining equations \eqref{constraint-rels}.  For this case, the order $\lambda^3$  equation is satisfied, provided the following  relation holds
\begin{align}
\begin{split}
\frac{\eta_3}{2} &= -\,\frac{b_1-b_2}{a_1-a_2}\,\xi_3
\,.
\label{lambda3-rel}
\end{split}
\end{align}
After a short calculation, one finds that the order $\lambda^2$ equation is satisfied, provided the additional two relations hold 
\begin{align}
\begin{split}
\frac{\eta_3}{2} &= -\,\frac{a_1-a_2}{b_1-b_2}\,\xi_3
\,,\\
\xi_3(\eta_1-\eta_2)&= \frac{\eta_3}{2} (\xi_1-\xi_2)
\,.
\label{lambda2-rel}
\end{split}
\end{align}
Upon rearrangement, equations \eqref{lambda3-rel}--\eqref{lambda2-rel} combine into 
\begin{align}
\begin{split}
\left(\frac{a_1-a_2}{b_1-b_2}\right)^2 &= 1\,.
\label{lambda2+3-rel}
\end{split}
\end{align}
for nonzero $\eta_3$ and $\xi_3$. Hence, equations \eqref{lambda3-rel}--\eqref{lambda2-rel} imply
\begin{align}
\begin{split}
\frac{\eta_3}{2} &= \pm \xi_3
\,,\\
\eta_1-\eta_2  &= \pm (\xi_1-\xi_2)
\,,\\
&\hbox{since}\quad
\frac{a_1-a_2}{b_1-b_2} =  \pm1
\,.
\end{split}
\label{3lin-rel}
\end{align}
Thus, the two-particle Toda $G$-strand has a zero curvature representation (or Lax pair) and is  completely integrable, provided the conditions \eqref{3lin-rel} hold.

In this case, we may use the two linear relations in \eqref{3lin-rel} among the six variables $\{\eta_1,\eta_2,\eta_3,\xi_1,\xi_2,\xi_3\}$ to simplify the original equations \eqref{Ham-eqns}.

\paragraph{The ZCR test for integrability}
By reordering equations \eqref{Ham-eqns} we may write them as three pairs of nonlinear wave equations in characteristic form,
\begin{align}
\begin{split}
& \partial_t(\e_1+\e_2) - \partial_s(\xi_1+\xi_2) = 0 \,, \\
& \partial_t(\xi_1 + \xi_2) - \partial_s(\e_1 + \e_2) =0 \,, \\
& \partial_t(\e_1-\e_2) - \partial_s(\xi_1-\xi_2) = 0 \,, \\
& \partial_t(\xi_1-\xi_2) - \partial_s(\e_1 - \e_2) = 4 \xi_3^2 - \e_3^2 \,, \\
& \partial_t\xi_3 -  \partial_s\frac{\e_3}{2}  = 
- \xi_3(\xi_1 - \xi_2) \left[ 1 - \frac{\e_3}{2\xi_3}\frac{\e_1 - \e_2}{\xi_1 - \xi_2} \right]
, \\
& \partial_t\frac{\e_3}{2} - \partial_s\xi_3 =
 -\xi_3(\xi_1 - \xi_2) \left[ \frac{\e_1 - \e_2}{\xi_1 - \xi_2} - \frac{\e_3}{2\xi_3}\right]
.
\end{split}
\label{Ham-eqns1}
\end{align}
Under the linear relations imposed by the integrability conditions \eqref{3lin-rel}, we see that the nonlinearities on the right-hand sides of equations \eqref{Ham-eqns1} all \emph{vanish} and the system becomes \emph{linear}.

This result has proved the following theorem.
\begin{theorem}[Cancellation of the nonlinearity under the ZCR test for integrability]\rm$\,$\\
The determining equations in \eqref{3lin-rel} provide the submanifold on which the two-particle T-Strand PDE system \eqref{Ham-eqns1} in 1+1 dimensions admits a ZCR. These (linear) determining equations imply that the (quadratic) nonlinear terms \emph{vanish} in the T-Strand PDEs.
\end{theorem}

\section{Three-particle T-Strands}\label{3PT-Strand-sec}

To write the three-particle T-Strand equations on the semidirect product Lie algebra of the Lie group $S\circledS T$ in \eqref{Toda-sym-3part}, one may follow the same path as in section \ref{2PT-strand-sec} for the two-particle case. This means starting from the EP formulation of the three-particle Toda lattice ODEs as a basis for the EP formulation of the corresponding PDEs for three-particle T-Strands. One then inverts the Legendre transform to express the EP equations and their compatibility equations all in terms of the Lie algebraic variables, $\xi,\eta\in \mathfrak{s}\circledS\mathfrak{t}$. 

We write the three-particle T-Strand equations on the semidirect product Lie algebra in component form. First there are five EP equations determined from the five components of $({\rm ad}^*_\xi\eta)$ in \eqref{ad-star-3part},
\begin{align}
\begin{split}
&\partial_t\xi_1 - \partial_s\e_1 
= ({\rm ad}^*_{\xi}\e)_1
= \xi_1(\xi_4-\xi_3)  -  \e_1(\e_4-\e_3)\,,  \\
&\partial_t\xi_2 - \partial_s\e_2  
= ({\rm ad}^*_{\xi}\e)_2
=  \xi_2(\xi_5-\xi_4) - \e_2(\e_5-\e_4) \,, \\
& \partial_t\xi_3 -  \partial_s\e_3  
= ({\rm ad}^*_{\xi}\e)_3
= \frac{1}{2} (\xi_1^2 - \e_1^2) \,,\\
&\partial_t\xi_4 - \partial_s\e_4 
= ({\rm ad}^*_{\xi}\e)_4
= \frac{1}{2}(\xi_2^2-\xi_1^2)
- \frac{1}{2}(\e_2^2-\e_1^2)
\,, \\
&\partial_t\xi_5 - \partial_s\e_5  
= ({\rm ad}^*_{\xi}\e)_5
= -\,\frac{1}{2} (\xi_2^2 - \e_2^2)
\,.
\label{Ham-eqns-3P} 
\end{split}
\end{align}
Then there are five compatibility equations determined from the components of $({\rm ad}_\xi\eta)$ in \eqref{ad-commut-3part},
\begin{align}
\begin{split}
& \partial_t\e_1 - \partial_s\xi_1 
= -({\rm ad}_{\xi}\e)_1 
=  \xi_1(\eta_4-\eta_3) - \eta_1 (\xi_4-\xi_3)\,, \\
& \partial_t\e_2 - \partial_s\xi_2 = -({\rm ad}_{\xi}\e)_2
= \xi_2(\eta_5-\eta_4) - \eta_2(\xi_5-\xi_4) \,, \\
& \partial_t\e_3 - \partial_s\xi_3 = -({\rm ad}_{\xi}\e)_3 = 0 \,, \\
& \partial_t\e_4 - \partial_s\xi_4 = -({\rm ad}_{\xi}\e)_4 = 0 \,, \\
& \partial_t \e_5 - \partial_s\xi_5 = -({\rm ad}_{\xi}\e)_5 = 0 \,.  
\label{comp-eqns-3P} 
\end{split}
\end{align}

\subsection{Hyperbolicity of the three-particle T-Strand} 
\begin{theorem}\label{charactform}\rm
Equations \eqref{Ham-eqns-3P}--\eqref{comp-eqns-3P} for three-particle T-Strands may be written as a ten-dimensional symmetric hyperbolic system with characteristic speeds $c$ given by $c^2=1$.
\end{theorem}

\begin{proof}
The three-particle T-Strand equations \eqref{Ham-eqns-3P}--\eqref{comp-eqns-3P} may be rearranged into the following symmetric hyperbolic form,
\begin{equation}
\frac{\partial}{\partial t} 
\begin{bmatrix} \boldsymbol{\xi} \\  \boldsymbol{\eta} \end{bmatrix} 
-
\begin{bmatrix} 
\mathbf{0}_{5\times5} & \mathbf{Id}_{5\times5}
\\
\mathbf{Id}_{5\times5} & \mathbf{0}_{5\times5}
\end{bmatrix} 
\frac{\partial}{\partial s} 
\begin{bmatrix} \boldsymbol{\xi} \\  \boldsymbol{\eta} \end{bmatrix} 
= RHS
\end{equation}
where $RHS$ is the right-hand side of equations \eqref{Ham-eqns-3P}--\eqref{comp-eqns-3P}. 
The characteristic polynomial of this system has roots given by $(c^2-1)^5$=0.
Thus, the characteristic speeds $c$ are given by $c^2=1$.
\end{proof}

\subsection{ZCR test for integrability of the three-particle T-Strand} 
\begin{proposition} [Lax pair for Euler--Poincar\'e] \rm
There exists a ZCR for the three-particle T-Strand Euler--Poincar\'e equation \eqref{Ham-eqns-3P}, namely
\begin{equation*}
\partial_t M_1 -\partial_s N_1 + \left[M_2,M_1\right] - \left[N_2,N_1\right] = 0
\end{equation*}
where
\begin{align*}
M_1
=
\begin{pmatrix}
 \xi_3 & \xi_1/2 & 0 
\\
\xi_1/2 & \xi_4 &  \xi_2/2
\\
0 & \xi_2/2 & \xi_5
\end{pmatrix}
,&\qquad
M_2
=
\begin{pmatrix}
 \xi_3 & 0 & 0 
\\
\xi_1 & \xi_4 & 0
\\
0 & \xi_2 & \xi_5
\end{pmatrix}
\\
N_1
=
\begin{pmatrix}
 \eta_3 & \eta_1/2 & 0 
\\
\eta_1/2 & \eta_4 &  \eta_2/2
\\
0 & \eta_2/2 & \eta_5
\end{pmatrix}
,&\qquad
N_2
=
\begin{pmatrix}
 \eta_3 & 0 & 0 
\\
\eta_1 & \eta_4 & 0
\\
0 & \eta_2 & \eta_5
\end{pmatrix}
\end{align*} 
\end{proposition}
\begin{proof}
The proof is a direct calculation. 
\end{proof}

\begin{proposition} [ZCR for compatibility condition]\rm 
It follows from the above definitions that these matrices satisfy the compatibility condition for the three-particle Toda Lattice G-Strands
\begin{equation*} 
\partial_t N_2 - \partial_s M_2 + \left[M_2,N_2\right] = 0 
\,,
\end{equation*}
although only on the submanifold $\eta_2\x_1=\x_2\eta_1$.
\end{proposition}
\begin{proof}
The proof is a direct calculation. 
\end{proof}

Keeping in mind the caveat that the compatibility condition for three-particle T-Strand holds only on the submanifold $\eta_2\x_1=\x_2\eta_1$, we continue with the calculation of the ZCR condition. 

\begin{theorem}[Cancellation of the nonlinearity under the ZCR test for integrability]\rm$\,$\\
The determining equations in \eqref{3lin-rel} provide the submanifold on which the three-particle T-Strand PDE system \eqref{Ham-eqns1} in 1+1 dimensions admits a ZCR. These (linear) determining equations imply that the (quadratic) nonlinear terms \emph{vanish} in the T-Strand PDEs.
\end{theorem}

\begin{proof} 
We define quadratic quantities in $\lambda$
\begin{equation*} 
L:=\lambda^2 A + \lambda M_1 - N_2 \qquad  M:=\lambda^2 B - \lambda N_1 + M_2 
\end{equation*}
The Lax pair for a system of PDEs is defined as:
\begin{equation} 
\partial_t L - \partial_s M = [L,M] 
\end{equation}
Inserting our definitions for L and M, and equating for powers of $\lambda$ yields:
\begin{align}
\begin{split}
\lambda^4 &: [A,B]=0 \\
\lambda^3 &: [A,N_1]+[B,M_1]=0 \\
\lambda^2 &: [A,M_2]+[N_2,B]-[M_1,N_1]=0 \\
\lambda^1 &: \partial_t M_1 -\partial_s N_1 + \left[M_2,M_1\right] - \left[N_2,N_1\right] = 0 \\
\lambda^0 &:  \partial_t N_2 - \partial_s M_2 + \left[M_2,N_2\right] = 0
\end{split}
\end{align}

We now use the proof of Theorem \ref{ZCR-thm} in order to derive the relations under which the system of equations (\ref{Ham-eqns}) is integrable.
\vspace{5mm} \\
The $\lambda^4$ equation is satisfied by choosing $A=diag(a_1,a_2,a_3)$ and $B:=diag(b_1,b_2,b_3)$.
\vspace{5mm} \\
For our system, a computation of matrix algebra yields that the order $\lambda^3$ equation is satisfied provided the following relations hold
\begin{align}
\begin{split}
-\eta_1&= \frac{b_1-b_2}{a_1-a_2}\,\x_1 \,, \\
-\eta_2&=\frac{b_2-b_3}{a_2-a_3}\,\x_2 \,, 
\label{lambda3-rel1}
\end{split}
\end{align}
We similarly find that the order $\lambda^2$ equation is satisfied provided the following relations hold 
\begin{align}
\begin{split}
\eta_1(\x_4-\x_3)&=\x_1(\eta_4-\eta_3)
\,. \\
\eta_2\x_1&=\x_2\eta_1
\,, \\
-\eta_1&=\frac{a_1-a_2}{b_1-b_2}\x_1
\,, \\
\eta_2(\x_5-\x_4)&=\x_2(\eta_5-\eta_4)
\,, \\
-\eta_2&=\frac{a_2-a_3}{b_2-b_3}\,\x_2  
 \,.
\label{lambda2-rel1}
\end{split}
\end{align}

By comparing the expressions obtained from solving for the ratios $\x_1/\eta_1$ and $\x_2/\eta_2$ in \eqref{lambda3-rel1} and \eqref{lambda2-rel1}, we see that the following relations must hold
\begin{equation} \frac{b_1-b_2}{a_1-a_2} =
\frac{b_2-b_3}{a_2-a_3}  = \pm 1 \end{equation}
Hence, from the other equations in \eqref{lambda2-rel}, we also have 
\begin{equation} 
\frac{\xi_1}{\eta_1}  = \frac{\xi_4-\xi_3}{\eta_4-\eta_3} = \mp 1\,,
\qquad
\frac{\xi_2}{\eta_2} = \frac{\xi_5-\xi_4}{\eta_5-\eta_4} = \mp 1
\,.
\end{equation}

Upon substituting the ZCR relations \eqref{lambda2-rel1} into the three-particle T-Strand equations \eqref{Ham-eqns-3P}, all of their nonlinearities cancel to zero. That is, only the linear part of the three-particle T-Strand system of equations \eqref{Ham-eqns-3P} survives after imposing the ZCR test for integrability. 
\end{proof}

\section{Conclusions}\label{conclude-sec}
This paper has formulated Toda lattice dynamics in terms of Hamilton's principle by using the Euler-Poincar\'e (EP) theory for Lie group invariant Lagrangians \cite{Po1901}. This EP formulation of Toda lattice dynamics has allowed it to be placed into the G-Strand framework, which generalized it from ODEs in time $t$ to T-Strand PDEs in 1+1 space-time $(s,t)$. The travelling wave solutions for this two-particle T-Strand system of nonlinear PDEs were studied and represented geometrically, as paths in $\mathbb{R}^3$ along intersections between two families of orthogonally aligned but off set elliptical cylinders. The two-particle T-Strand PDEs turned out to comprise a six-dimensional symmetric hyperbolic system with constant characteristic speeds. This, in turn, was a necessary condition for them to be endowed with a zero curvature representation (ZCR). The ZCR was computed, but for the two-particle T-Strand the determining conditions for the existence of the ZCR exactly cancelled the nonlinearity in the PDE system, resulting in a set of six uncoupled linear wave equations.  This was a bit surprising, since the strategy of extending the Lax pair for an integrable ODE system to an an integrable 1+1 PDE system had been successful in other cases \cite{HoIvPe2012}. One might have thought that the cancellation of nonlinearities had happened for the two-particle T-Strand because it did not have enough ``room'' for the interplay among it degrees of freedom to admit the additional determining conditions. However, the same phenomenon occurred for the three-particle T-Strand PDEs, in which the nonlinear terms were again exactly cancelled by the ZCR determining conditions.  

We conclude by stating the following open question  raised by this investigation and suggesting a possible route to its solution. Namely, when does the Euler-Poincar\'e variational strategy of the present paper fail or succeed in producing an integrable 1+1 PDE system from an integrable ODE system? As mentioned earlier, this approach had succeeded in formulating  integrable G-Strand PDEs by extending Euler's equations for both the rigid body ODEs on $SO(3)$ and the Bloch-Iserles ODEs on $Sp(2)$, both of which have Lax pairs \cite{HoIvPe2012}.  As we have found, the EP approach does produce interesting PDE systems for the two-particle and three-particle T-Strands. However, in the T-Strand case, the EP approach does not succeed in producing nonlinear systems that admit a ZCR and are thus completely integrable Hamiltonian systems. The reason for this may have been spotted in remark \ref{ad-dagger-rem} where we observed that on the Lie algebra $\mathfrak{s}\circledS\mathfrak{t}$ the quantity $({\rm ad}^*_\xi\mu)^T\in \mathfrak{g}$ \emph{cannot} be written as a matrix commutator with the matrix pairing, since 
\[
({\rm ad}^*_\xi\mu)^T=:{\rm ad}^\dagger_\xi(\mu^T)\ne -\,{\rm ad}_\xi(\mu^T)
\,.
\]
This means that for this Lie algebra and this pairing the operations ad and ad$^*$ are not both equivalent to matrix commutators. 
In contrast, one has ${\rm ad}^\dagger_\xi(\mu^T) = -\,{\rm ad}_\xi(\mu^T)$ for both of the Lie algebras $\mathfrak{so}(3)$ and $\mathfrak{sp}(2)$, whose corresponding EP equations have been shown to be completely integrable via the EP approach.

\subsection*{Acknowledgements}
We are grateful to A. M. Bloch, J. Carrillo de la Plata, L. Garcia, F. Gay-Balmaz, R. I. Ivanov, H. O. Jacobs, T. S. Ratiu, C. Tronci and J. Vankerschaver for timely discussions and suggestions. The work by DDH was partially supported by Advanced Grant 267382 FCCA from the European Research Council.


\begin{thebibliography} {}\rm

\bibitem{AdvMoVa2004}
Adler, M., van Moerbeke, P. and Vanhaecke, P. [2004]
\textit{Algebraic Integrability, Painlev\'e Geometry and Lie Algebras}
Springer, New York. 

\bibitem{BlBrRa1990}
Bloch, A. M., Brockett, R. W. and Ratiu, T. S. [1990]
A new formulation of the generalized Toda lattice equations and their fixed point analysis via the momentum map,
\textit{Bulletin of the Amer. Math. Soc.} 23: (2) 477--485.

\bibitem{DzVo1980} 
Dzyaloshinskii, I. E. and Volovick, G. E. [1980]
Poisson brackets in condensed matter systems,
\textit{Ann. Phys.} {\bf 125}: 67--97.

\bibitem{ElGBHoPuRa2010} 
Ellis, D. C. P., Gay-Balmaz, F.,  Holm, D. D., Putkaradze, V. and Ratiu, T. S. 
[2010] Symmetry reduced dynamics of charged molecular strands,
\textit{Arch. Rational Mech. Anal.}  {\bf 197}: (3)  811--902.

\bibitem{FPU1955} 
E. Fermi, J. Pasta, and S. Ulam, Studies of the Nonlinear Problems, I, Los Alamos Report
LA-1940, (1955), later published in \emph{Collected Papers of Enrico Fermi}, ed. E. Segre, Vol. II
(University of Chicago Press, 1965) p.978; also reprinted in \emph{Nonlinear Wave Motion}, ed.
A. C. Newell, Lecture Notes in Applied Mathematics, Vol. 15 (AMS, Providence, RI, 1974),
also in \emph{Many-Body Problems}, ed. D. C. Mattis (World Scienti�c, Singapore, 1993).

\bibitem{Fl1974a} Flaschka, H. [1974]
The Toda lattice. II. Existence of integrals,
{\it Phys. Rev. B}, {\bf 9}, 1924--1925. 

\bibitem{Fl1974b} Flaschka, H. [1974]
On the Toda lattice. II,
{\it Progr. Theoret. Phys.}, {\bf 51} (3) 703--716.

\bibitem{FeStTu1973} Ford, J. et al [1973]
On the integrability of the Toda lattice,
{\it Prog. Theoret. Phys.}, {\bf 50} 1547.

\bibitem{Ge12} Gerdjikov, V. S. [2012] Riemann-Hilbert Problems with canonical normalization
 and families of commuting operators, {\it Pliska Stud. Math. Bulgar.} 21, 201-216.
 \url{http://arxiv.org/pdf/1204.2928v1.pdf}

\bibitem{GeIvKu80} Gerdjikov, V. S., Ivanov, M. I., Kulish, P. P.
[1980] Quadratic pencils and nonlinear equations. (Russian) {\it
Teoret. Mat. Fiz.} 44, no. 3, 342-357.

\bibitem{GBRa2009} 
Gay-Balmaz, F. and Ratiu, T. S. [2009]
The geometric structure of complex fluids,
\textit{Adv. Appl. Math.} 42 (2): 176--275.

\bibitem{Ho2011} Holm, D.D. [2011]
{\em Geometric Mechanics. Part I: Dynamics and Symmetry}.
2nd ed. Imperial College Press.

\bibitem{He1974} H\'enon, M. [1974] 
Integrals of the Toda lattice,
{\it Phys. Rev. B}, {\bf 9}, 1921-1923.

\bibitem{Ho2002}   
Holm, D. D. [2002]
Euler-Poincar\'e dynamics of perfect complex fluids, 
in \textit{Geometry, Dynamics
and Mechanics: 60th Birthday Volume for J.E. Marsden}. 
P. Holmes, P. Newton, and A. Weinstein, eds., Springer-Verlag, pp 113--167.

\bibitem{HoIvPe2012}  
Holm, D. D., Ivanov, R. I., Percival, J. R. [2012]
G-Strands. {\it J. Nonlinear Science} 22(4): 517--551. 

\bibitem{HoKu1984} Holm, D. D.  and Kupershmidt, B. A. [1984] Yang-Mills magnetohydrodynamics: Nonrelativistic theory, 
\textit{Phys. Rev. D} 30: 2557--2560.

\bibitem{HoKu1988} Holm, D. D.  and Kupershmidt, B. A. [1988] 
The analogy between spin glasses and Yang-Mills fluids, 
\textit{J. Math. Phys.} 29: 21--30.

\bibitem{HoMaRa1998}
     \textrm{Holm, D. D., Marsden, J. E. and Ratiu, T. S.} 
[1998] 
     {The Euler--Poincar\'e equations and semidirect products
     with applications to continuum theories.}
     \textit{Adv. Math.} \textbf{137}, 1--81.

\bibitem{HoScSt2009}
     \textrm{Holm, D. D., Schmah, T. and Stoica, C.} [2009] 
          {\it Geometric Mechanics and Symmetry: 
From Finite to Infinite Dimensions}.
Oxford: Oxford University Press. 

\bibitem{KaNe78} Kaup, David J.; Newell, Alan C. [1978]
An exact solution for a derivative nonlinear Schr\"odinger
equation. {\it J. Math. Phys.} 19, no. 4, 798-801.

\bibitem{KuMi77} Kuznetsov, E. A.,  Mihailov, A. V. [1977]
The complete integrability of the two-dimensional classical Thirring model. (Russian)
{\it Teoret. Mat. Fiz.} 30 (1977), no. 3, 303-314.

\bibitem{Lax68} Lax, P. D. [1968]
Integrals of nonlinear equations of
evolution and solitary waves,
{\em Comm. Pure Appl. Math.} {\bf 21}, 467--490.

\bibitem{Ma1976} Manakov, S. V. [1976] A remark on the integration of the Euler equations for the dynamics of an n-dimensional rigid body,
{\it Functional Anal. Appl.} 10, 328--329.
Translated from {\it Funktsional. Anal. i Ego Prilozhen.}10:4, 93--94.

\bibitem{Po1901} 
     \textrm{Poincar\'e, H.} [1901]
     {Sur une forme nouvelle des \'{e}quations de la m\'{e}canique}.
     {\it C. R. Acad. Sci.\/} {\bf 132}, 369--371.
     
\bibitem{Mo1976} \textrm{Moser, J.} [1976]
J. Moser, Finitely many mass points on the line under the influence of an exponential potentialÑan integrable system, 
{\it Lecture Notes in Physics}, vol. 38, pp. 97--101.

\bibitem{Toda1970} Toda, M. [1970]
Waves in nonlinear lattice,
{\it Supplement of the Progress of Theoretical Physics}, {\bf 46} 174-200

\bibitem{ZaMa1973} V.E. Zakharov, S.V. Manakov, O rezonansnom vzaimodeistvii volnovykh paketov v nelineinykh sredakh, {\it Pis�ma v ZhETF}, 18 (7), 413-417 (1973) [V.E. Zakharov, S.V. Manakov, Resonant interaction of wave packets in nonlinear media, {\it JETP Lett.}, 18 (7), 243-245 (1973)].

\bibitem{ZaMaNoPi1984}
Zakharov, V. E., Manakov, S. V., Novikov, S. P. and Pitaevsky, L. P. [1984] {\it Theory of Solitons}, Plenum (New York).

\bibitem{ZaMi1980} Zakharov, V.E. and Mikhailov, A.V. [1980] On the integrability of classical spinor models in two-dimensional space-time, {\it Commun. Math. Phys.} {\bf 74}, 21--40.


\end{thebibliography}
\end{document}